\documentclass[pra,aps,twocolumn,superscriptaddress,10pt]{revtex4-1}

\usepackage[dvipdfmx]{graphicx}
\usepackage{color}
\usepackage{amsmath,amssymb,amsthm,amsfonts}

\newtheorem{theorem}{Theorem}

\newtheorem{lemma}{Lemma}

\newtheorem{proposition}{Proposition}
\newtheorem{remark}{Remark}

\newcommand{\cH}{{\cal H}}

\newcommand{\ket}[1]{|#1\rangle}
\newcommand{\Tr}{{\textrm Tr}\,}

\newcommand{\sH}{{H}}
\newcommand{\sX}{{X}}
\newcommand{\sZ}{{Z}}
\newcommand{\sY}{{Y}}

\newcommand{\sA}{{A}}

\newcommand{\sD}{{D}}
\newcommand{\sU}{{U}}

\newcommand{\bF}{\mathbb{F}}

\def\Label#1{\label{#1}\ [\ \text{#1}\ ]\ }
\def\Label{\label}

\def\QED{\mbox{\rule[0pt]{1.5ex}{1.5ex}}}
\def\endproof{\hspace*{\fill}~\QED\par\endtrivlist\unskip}

 \newenvironment{proofof}[1]{\vspace*{5mm} \par \noindent
         \quad\textit{ Proof of #1:\hspace{2mm}}}{\endproof}

\begin{document}
\title{Self-guaranteed measurement-based quantum computation}
\author{Masahito Hayashi}\email{masahito@math.nagoya-u.ac.jp}
\affiliation{Graduate School of Mathematics, Nagoya University, Furocho, Chikusa-ku, Nagoya 464-8602, Japan}
\affiliation{Centre for Quantum Technologies, National University of Singapore, 3 Science Drive 2, Singapore 117543}

\author{Michal Hajdu\v{s}ek}\email{cqtmich@nus.edu.sg}
\affiliation{Centre for Quantum Technologies, National University of Singapore, 3 Science Drive 2, Singapore 117543}

\begin{abstract}
In order to guarantee the output of a quantum computation, we usually assume that the component devices are trusted.
However, when the total computation process is large, it is not easy to guarantee the whole system when we have scaling effects, unexpected noise, or unaccounted correlations between several subsystems.
If we do not trust the measurement basis nor the prepared entangled state, we do need to be worried about such uncertainties.
To this end, we proposes a ``self-guaranteed'' protocol for verification of quantum computation under the scheme of measurement-based quantum computation where no prior-trusted devices (measurement basis nor entangled state) are needed.
The approach we present enables the implementation of verifiable quantum computation using the measurement-based model in the context of a particular instance of delegated quantum computation where the server prepares the initial computational resource and sends it to the client who drives the computation by single-qubit measurements.
Applying self-testing procedures we are able to verify the initial resource as well as the operation of the quantum devices, and hence the computation itself.
The overhead of our protocol scales as the size of the initial resource state to the power of 4 times the natural logarithm of the initial state's size.
\end{abstract}

\maketitle

\section{Introduction}
Quantum computation offers a novel way of processing information and promises solution of some classically intractable problems ranging from factorization of large numbers \cite{Shor1994} to simulation of quantum systems \cite{Lloyd1996}.
However, as quantum information processing technologies improve the performance of quantum devices composed of ion traps and superconducting qubits \cite{Harty2014,Barends2014}, a natural question arises; ``How can we guarantee the computation outcome of a prepared quantum computation machine?''
The solution of this problem is strongly desired in the context of characterization, verification and validation  of quantum systems (QCVV), which is actively addressed in recent studies \cite{QCVV1,QCVV2}.
For problems such as factorization, this does not present an issue as verification takes the form of simple multiplication of numbers.
However, we cannot deny a possibility that the constructed quantum device suffers from unexpected noise or unaccounted correlations between several subsystems resulting from our insufficient experimental control when implementing the quantum computer.
That is, we need to guarantee (verify) the outcome without any noise model.
This task is called the verification of quantum computation \cite{Reichardt2013nature,hajduvsek2015device,McKague2016,fitzsimons2015post,Hangleiter2017,McCutcheon2016,Aolita2015,mantri2017flow,Barz2013,Broadbent2009,fitzsimons2012unconditionally,fitzsimons2016private,Gheorghiu2015,Morimae2014verification,gheorghiu2017review}.

The concept of verifying a quantum computation is quite different from quantum error correction.
In quantum error correction we start with a noise model that can adversely affect the computation and devise quantum codes to counteract this noise provided its strength remains below a certain threshold. In verification of quantum computation we do not make any assumptions about the noise.
The prepared states and measurement devices may be behaving ideally or they may be affected by noise.
The goal of verification is to ascertain whether the quantum states and measurement devices behave closely according to specifications and how this deviation affects the output of the computation, \textit{without} assuming any noise model.
This is crucial from an experimental point of view as it allows us to test quantum devices and guarantee their reliable operation.

For this purpose, the verification of quantum computation needs to satisfy the following two requirements. 
One is \textit{detectability} which means that  if the state or the measurement device is far from the ideal one, we reject it with high probability. 
In this stage, no assumption on the underlying noise model should be made. 
The other is \textit{acceptability} which means that the ideal state and the ideal measurement device can pass the test with high probability. 
Both requirements are needed to characterize performance of test in statistical hypothesis testing \cite{Lehmann2008,fujii2016verifiable}.

We need to clarify whether we have already verified the device or not.
To address this issue, a device is called \textit{trusted} when we have already verified it.
Otherwise, it is called \textit{untrusted}.
This task may seem daunting at first, particularly when considered in the context of quantum circuit model \cite{Nielsen2000}.
In this model, the computation takes form of a sequence of local and multi-local unitary operations applied to the quantum state resulting in a quantum output that is finally measured out to yield the classical result of the computation.
In order to verify the correctness of the output it would appear that one needs to keep track of the entire dynamics, effectively classically simulating the quantum computation.
This can of course be achieved only for the smallest of quantum systems due to the exponential increase in the dimensionality of the Hilbert space with increasing system size.
Measurement-based model of quantum computation (MBQC), is equivalent to the quantum circuit model but uses non-unitary evolution to drive the computation \cite{Raussendorf2001,Raussendorf2003,Briegel2009}.
In this model, the computation begins with preparation of an entangled multi-qubit resource state and proceeds by local projective measurements on this state that use up the initial entanglement.
In order to implement the desired evolution corresponding to the unitary from the circuit model, the measurements must be performed in an adaptive way where future measurement bases depend on previous measurement outcomes which imposes a temporal ordering on the measurements.

The initial proposal of MBQC in \cite{Raussendorf2001} considered a cluster state as the initial state \cite{briegel2001persistent} and measurements in the $X$-$Y$ plane of the qubit's Bloch sphere at an arbitrary angle along with $Z$ measurements.
It has been recently shown that $Z$ measurements are in fact not necessary \cite{mantri2016universality}.
We consider measurements of $X$, $Z$ and $X\pm Z$ that are approximatly universal when paired with a triangular lattice as the initial resource state \cite{mhalla2013graph}.
For trusted measurement devices, the computation outcome can be guaranteed only by verifying the initial entangled multi-qubit resource state \cite{Hayashi2015} using stabilizer measurements.
However, for untrusted devices this method alone is not sufficient.

Our task is guaranteeing the computation outcome \textit{without} trusting the measurement devices as well as the initial entangled resource state.
To achieve this, we employ self-testing techniques to guarantee prepared states as well as to certify the operation of quantum devices.
Self-testing, originally proposed in \cite{Mayers1998,Mayers2004}, is a statistical test that compares measured correlations with the ideal ones and based on the closeness of these two cases draws conclusions whether the real devices behave as instructed under a particular definition of equivalence.
In any run of the computation we assume that the prepared physical states and devices are untrusted and therefore need testing.
Self-testing does not make any artificial assumptions about the Hilbert space structure of the devices or the measurement operators corresponding to classical outcomes observed.

To achieve verification of quantum computation, we need to establish a self-test for a triangular graph state as well as measurements mentioned in the above paragraph.
McKague proposed a self-testing procedure for such a graph state in \cite{McKague2011} along with measurements in the $X$-$Z$ plane.
However, this method requires many copies of the $n$-qubit graph state scaling as $O(n^{22})$ and therefore is not possible with current or near-future quantum technologies.

In this paper, with feasible experimental realization in mind, we propose a self-testing procedure for a triangular graph state along with measurements of $X$, $Z$ and $X\pm Z$.
One of the main differences between MBQC and the quantum circuit model is the clear split between preparation of the initial entangled resource state and the computation itself.
This property suggests a natural approach to guaranteeing the outcome of the computation by splitting the verification process into two parts.
Firstly, we have to verify the initial entangled multi-qubit resource state.
Secondly, we guarantee the correct operations of the measurement devices that drive the computation.
To realize this approach, we begin by introducing a protocol that reduces self-testing for a triangular graph state to a combination of self-tests for a Bell state.

Original proposals of Mayers and Yao \cite{Mayers1998,Mayers2004} have considered self-testing of a Bell state.
The method of \cite{McKague2010} is based on the Mayers-Yao test while \cite{McKague2012} discusses methods based on the Mayers-Yao test as well as the CHSH test.
These two approaches require relatively small number of measurement settings.
However, direct application of these methods to our protocol results in a huge number of required copies of the graph state.
To resolve this issue, we propose a different method for self-testing of a Bell state, which has better precision as previous methods.

\section{Self-testing of measurements based on two-qubit entangled state}
As the first step, we consider a self-testing protocol of local measurements on the untrusted system ${\cal H'}_1$ when the untrusted state $|\Phi'\rangle$ is prepared on the bipartite system ${\cal H}'_1\otimes{\cal H}'_2$.
The trusted state corresponding to $|\Phi'\rangle$ is $(|0,+\rangle + |1,-\rangle)/\sqrt{2}$. Note that even though the trusted system is a two-qubit state, we do not assume that either of the untrusted systems ${\cal H}'_1$ or ${\cal H}'_2$ are $\mathbb{C}^{2}$.
In the rest of our paper we denote untrusted states and operators with primes, such as $\ket{\psi'}$ and $X'$, in order to distinguish them from trusted states and devices which have no primes.
Our protocol satisfies the following requirements related to our self-testing protocol for three-colorable graph state.
\begin{description}
\item[(1-1)]
Identify measurements of $X_1$, $Z_1$ and $(X_1\pm Z_1)/\sqrt{2}$ within a constant error $\epsilon$.

\item[(1-2)]
Measure $X'_1$, $Z'_1$, $A(0)'_1$ and $A(1)'_1$ on the system ${\cal H}'_1$, where $A(i)'_1:=[X'_1+(-1)^iZ'_1]/\sqrt{2}$.

\item[(1-3)]
Measure only $X'_2$ and $Z'_2$ on the system ${\cal H}'_2$.

\item[(1-4)]
Prepare only $O(\delta^{-4})$ samples for the required precision level $\delta$, whose definition will be given latter.
\end{description}
Requirement \textbf{(1-1)} is needed for universal computation based on measurement-based quantum computation \cite{mhalla2013graph}.
Three-colorable graph states can be partitioned into three subsets of non-adjacent qubits.
In the rest of our paper, we refer to one of these subsets as black qubits (B), the second subset is referred to as white qubits (W) and the final subset are red qubits (R).
To realize the self-guaranteed MBQC of $n$-qubit three-colorable graph state with resource size $O(n^4 \log n)$, we need the requirement \textbf{(1-4)}.
Indeed, McKague et al. \cite{McKague2012} already gave a self-testing protocol for the Bell state.
However, their protocol requires resource size that scales as $O(n{}^8)$ (Remark 1 of Appendix \ref{app:Bell}).

\begin{figure}[t]
	\includegraphics[scale=0.70]{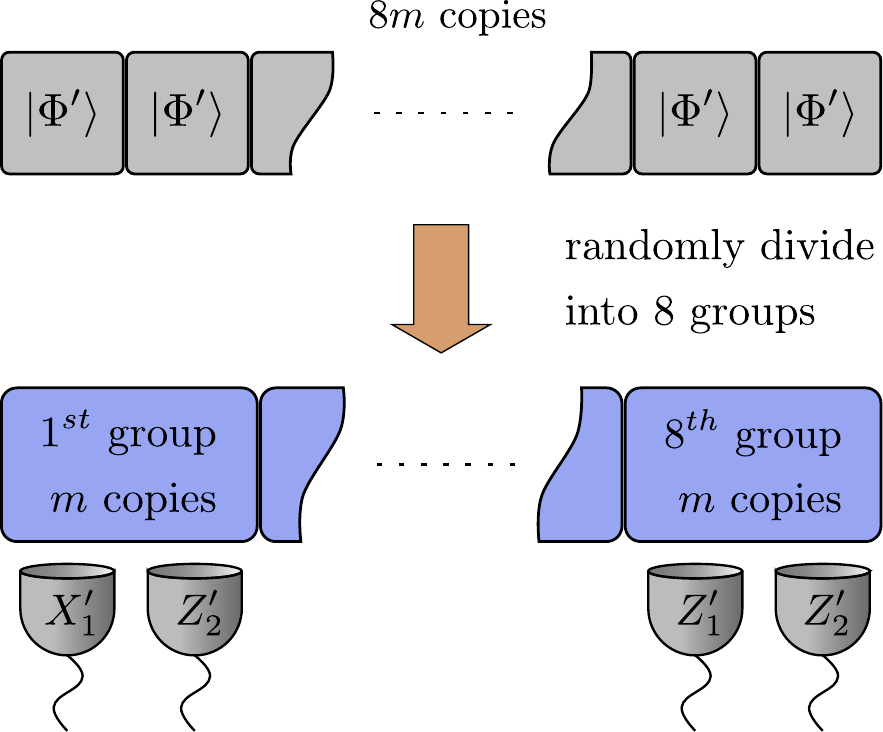}{}{}
	\caption{Representation of the self-testing procedure for the state $(\ket{0,+}+\ket{1,-})/\sqrt{2}$.
	We prepare $8m$ copies of this state which are then randomly divided into 8 groups.
	Each group is measured as described in \textbf{(2-3)} and \textbf{(2-4)}.
	There are 4 measurement settings for system $\mathcal{H}'_1$ and 2 measurement settings for system $\mathcal{H}'_2$.
	Each group is measured by one device acting on system $\mathcal{H}'_1$ and one device acting on system $\mathcal{H}'_2$.}
	\label{fig:figure1}
\end{figure}

The self-testing procedure is illustrated in FIG.~\ref{fig:figure1}.
We prepare $8m$ copies of the initial state and split them randomly into 8 groups that are 
then measured to test the correlations.
The procedure is described as follows and is denoted by Test {\bf (2)}:
\begin{description}
\item[(2-1)]
Prepare $8m$ states $|\Phi'\rangle$.

\item[(2-2)]
Randomly divide $8m$ blocks into 8 groups, in which, the 1st - 8th groups are composed of $m$ blocks.

\item[(2-3)]
Measure
$X'_1$, $Z'_1$,
$A(0)'_1$, $A(0)'_1$, $A(1)'_1$, $A(1)'_1$,
$X'_1$, $Z'_1$ on the system ${\cal H}'_1$
for the 1st - 8th groups.

\item[(2-4)]
The corresponding measurements on ${\cal H}'_2$ for the 8 groups are
$Z'_2$, $X'_2$,
$Z'_2$, $X'_2$, $Z'_2$, $X'_2$,
$X'_2$, and $Z'_2$.

\item[(2-5)]
Based on the above measurements, we check the following 5 inequalities for 8 average values:
\begin{align}
&Av [X'_1 Z'_2] = 1, \quad Av [Z'_1 X'_2] = 1, \label{e1}\\
&Av [A(0)'_1(Z'_2+X'_2)] \ge \sqrt{2}- \frac{c_1}{\sqrt{m}},\label{e2}\\
&Av [A(1)'_1(Z'_2-X'_2)]\ge \sqrt{2}- \frac{c_1}{\sqrt{m}},\label{e3}\\
&\big|Av [X'_1 X'_2+ Z'_1 Z'_2]\big| \le \frac{c_1}{\sqrt{m}}.\label{e4}
\end{align}
Here, for example, the average value in (\ref{e2}) is calculated from the outcomes of the 3rd and 4th groups.
\end{description}

This leads to the following theorem, which is shown in Appendix \ref{app:Bell}.
\begin{theorem}
Given a significance level $\alpha$ and an acceptance probability $\beta$, there exists a pair of positive real numbers $c_1$ and $c_2$ satisfying the following condition.
If the state $(|0,+\rangle + |1,-\rangle)/\sqrt{2}$ and measurement are prepared with no error, 
Test {\bf (2)} of the above $c_1$ is passed with probability $\beta$.
Once Test {\bf (2)} of the above $c_1$ is passed, 
we can guarantee, with significance level $\alpha$, that there exists an isometry $U:{\cal H'}_1\to  {\cal H}_1$ such that
\begin{align}
\| U X'_1 U^\dagger - X_1 \| & \le \delta, ~
\| U A(0)'_1 U^\dagger - A(0)_1 \| \le \delta ,\label{e5}\\
\| U Z'_1 U^\dagger - Z_1 \| & \le \delta, ~
\| U A(1)'_1 U^\dagger - A(1)_1 \| \le \delta ,\label{e6}
\end{align}
where $\delta:=c_2 m^{-1/4}$, which is called the required precision level.
\end{theorem}
Note that the significance level $\alpha$ is the maximum passing probability when one of the conditions in \textbf{(2-5)} does not hold \cite{Lehmann2008}.
The acceptance probability $\beta$ is also called the power of the test in hypothesis testing and is the probability to accept the test in the ideal case.
To satisfy the detectability and the acceptability, $\alpha$ and $\beta$ are chosen to be constants close to $0$ and $1$, respectively, which leads to their trade-off relation.
In this way, we can show how the measurements forming an approximately universal set for MBQC can be certified using a two-qubit state.
Now we proceed to extend this scheme to three-colorable states of arbitrary size.

\section{Self-testing of a three-colorable graph state}
Now, we give a self-testing for a three-colorable graph state $|G'\rangle$, composed of the black part (B), the white part (W), and the red part (R), whose total number of qubits is $n$.
Our protocol needs to prepare $ c m$ samples of the state $|G'\rangle$, where $m$ is $O(n^4 \log n)$, where the constant $c$ depends on the structure of the graph $G$.

To specify it, we introduce three numbers $l_B$, $l_W$, and $l_R$ for a three-colorable graph $G$.
Consider the set $S_B:=\{1, \ldots, n_B\}$ of black sites, the set $S_W:=\{1, \ldots, n_W\}$ of white sites, and the set $S_R:=\{1, \ldots, n_R\}$ of red sites.
We denote the neighborhood of the site $i$ by $N_i \subset S_W\cup S_R$.
We divide the sites $S_B$ into $l_B$ subsets $S_{B,1}, \ldots, S_{B,l_B}$ such that $N_i \cap N_j = \emptyset$ for $i \neq j \in S_{B,k}$ for any $k=1, \ldots, n_B$.
That is, elements of $S_{B,k}$ have no common neighbors, which is called the non-conflict condition.
We choose the number $l_B$ as the minimum number satisfying the non-conflict condition.
We also define the numbers $l_W$ and $l_R$ for the white and red sites in the same way.
In FIG.~\ref{fig:figure2}, we show that for a triangular graph $l_B, l_W, l_R\leq3$.
Based on this structure, testing of measurement devices on each site on $S_{B,k}$ can be reduced to the two-qubit case as follows,
\begin{description}
\item[(3-1)]
Prepare $8 m$ states $|G'\rangle$.

\item[(3-2)]
Measure $Z'$ on all sites of $S_B\setminus S_{B,k}$ for all copies.
Then, apply $Z'$ operators on the remaining sites to correct for the $Z'$ measurement depending on the outcomes.

\item[(3-3)]
For all $i \in S_{B,k}$, choose a site $j_i \in N_i$.
Then, measure $Z'$ on all sites of $S_W\setminus \{j_i\}_{i \in S_{B,k}}$ for all copies.
Apply $Z'$ operators on the remaining sites to correct for the $Z'$ measurements depending on the outcomes.

\item[(3-4)]
For ideal devices, the resultant state should be $\otimes_{i \in S_{B,k}}|\Phi'\rangle_{i j_i}$.
Finally, apply the above self-testing procedure to all of $\{|\Phi'\rangle_{i j_i}\}_{i \in S_{B,k}}$.
\end{description}
Since the above protocol verifies the measurement device on black sites in $S_{B,k}$, we call it B-protocol with $S_{B,k}$.
We define W-protocol with $S_{W,k}$ and R-protocol with $S_{R,k}$ in the same way.

\begin{figure}[t]
	\includegraphics[scale=1.5]{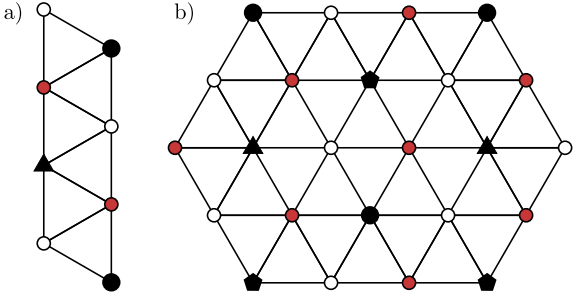}
	\caption{Two examples of three-colorable graphs with their black vertices partitioned into subsets $\{S_{B,k}\}_{k=1}^{l_{B}}$.
	For the small graph in a), we can see that $l_B\leq2$.
	The triangular lattice in b) is still three-colorable but has $l_B\leq3$.
	It can be readily checked that this partitioning satisfies the non-conflict condition since elements of the same partition $S_{B,j}$ do not share any common neighbors.
	White and red vertices may be partitioned in the same way which means that $l_W, l_R\leq3$.}
	\label{fig:figure2}
\end{figure}

Choosing $c_3$ to be $3+8(l_B+l_W+l_R)$, we propose the following self-testing protocol, which is denoted by Test \textbf{(4)}:
\begin{description}
\item[(4-1)]
Prepare $c_3 m+1$ copies of $n$-qubit state $|G'\rangle$.

\item[(4-2)]
Randomly divide the $c_3 m + 1$ copies into $c_3+1$ groups.
The first $c_3$ groups are composed of $m$ copies and the final group is composed of a single copy.

\item[(4-3)]
Measure $Z'$ on the black and white sites and $X'$ on the red sites for the 1st group and check that the outcome of $X'$ measurements
is the same as predicted from the outcomes of $Z'$ measurements.

\item[(4-4)]
Repeat the above stabilizer test on the 2nd group but measure white and red sites in the $Z'$ basis and black sites in $X'$ basis.

\item[(4-5)]
Repeat the above stabilizer test on the 2nd group but measure red and black sites in the $Z'$ basis and white sites in $X'$ basis.

\item[(4-6)]
Run the B-protocol with $S_{B,k}$ for the $4+8(k-1)$-th - $3+8 k$-th groups.
Repeat this protocol for $k=1, \ldots, l_B$.

\item[(4-7)]
Run W-protocol and R-protocol with $S_{W,k}$ and $S_{R,k}$ for the $4+8(l_B+k-1)$-th - $2+8(l_B+k)$-th and the $4+8(l_B+l_W+k-1)$-th - $2+8(l_B+l_W+k)$-th groups as in Step \textbf{(4-6)}, respectively.
Repeat this protocol for $k=1, \ldots, l_W$ and $k=1, \ldots, l_R$, respectively.

\end{description}
Steps \textbf{(4-3)}, \textbf{(4-4)}, and \textbf{(4-5)} perform the stabilizer test given in \cite{Hayashi2015} adapted to a triangular graph state which certifies the graph state $|G\rangle$.
For our self-testing, we need to guarantee local measurements of $X_1$, $Z_1$ and $(X_1\pm Z_1)/\sqrt{2}$ for all sites.
Since Test \textbf{(4)} utilizes B-protocol which in turn uses Test \textbf{(2)}, Test \textbf{(4)} depends on the parameter $c_1$ of Test \textbf{(2)}.

For acceptability, we need to pass Test \textbf{(2)} in all sites, i.e., $n$ qubits.
Hence, as shown in Appendix \ref{app:Th2}, to realize an acceptance probability $\beta$ close to $1$, we need to choose $c_1$ to be $c_4 (\log n)^{1/2}$ with a certain constant $c_4$, which leads to the following theorem.

\begin{theorem}\label{T2}
Given a significance level $\alpha$ and an acceptance probability $\beta$, there exists a pair of positive real numbers $c_2$ and $c_4$ satisfying the following condition.

If the state $|G\rangle$ and our measurements are prepared with no error, 
Test \textbf{(4)} with $c_1=c_4 (\log n)^{1/2}$ is passed with probability $\beta$.
Once Test \textbf{(4)} with $c_1=c_4 (\log n)^{1/2}$ is passed, 
we can guarantee, with significance level $\alpha$, that there exists an isometry $U_i:{\cal H'}_i\to {\cal H}_i$ such that
\begin{align}
\big\| U_i X'_i U_i^\dagger - X_i \big\| ,~
\big\| U_i Z'_i U_i^\dagger - Z_i \big\| &\le \delta \label{H15-1}\\
\big\| U_i A(0)'_i U_i^\dagger - A(0)_i \big\| , ~
\big\| U_i A(1)'_i U_i^\dagger - A(1)_i \big\| &\le \delta \label{H15-2B}\\
\textrm{Tr} \big[\sigma (I-P'_1)\big], \text{Tr} \big[\sigma (I-P'_2)\big], \textrm{Tr} \big[\sigma (I-P'_3)] &\le \frac{\alpha}{m}\label{H15-3}
\end{align}
where $\delta:=c_2 (\frac{\log n}{m})^{1/4}$, $U:=U_1 \otimes \cdots \otimes U_n$, $\sigma$ is the resultant state on the final group, and $P'_1$, $P'_2$, $P_3'$ are POVM elements corresponding to pass of Steps \textbf{(4-3)}, \textbf{(4-4)}, and \textbf{(4-5)}.
\end{theorem}
Here, the conditions \eqref{H15-1}--\eqref{H15-2B} follow from Theorem 1
and the condition \eqref{H15-3} follows from a similar discussion for the stabilizer test given in \cite{Hayashi2015}.

\section{Certification of the computational result}
To guarantee the computational result, we need to guarantee that our computational operation is very close to the true operation based on Theorem \ref{T2}.
When $\{M_i\}_i$ is a POVM realized by an adaptive measurement on each site from $X$, $Z$, $A(0)$, and $A(1)$, as shown in Appendix \ref{app:errorPOVM}, Theorem \ref{T2} guarantees that
\begin{align}
\big\| U M'_i U^\dagger -M_i \big\| \le 8 n \delta ,\label{H15-4}
\end{align}
where $M_i$ is the ideal POVM.
This inequality can be shown by a modification of a virtual unitary protocol composed of a collection of unitaries on each site controlled by another trusted system \cite[Lemma 3.6]{McKague2016}.
Thus, as shown in Appendix \ref{app:errorInitial}, Eqs.~\eqref{H15-3} combined with the relationship between trace distance and fidelity \cite{Hayashi2014} and the above discussion lead to
\begin{align}
\big\| U \sigma U^\dagger -|G\rangle \langle G| \big\|_1
\le {6 n \delta+ \frac{3 \alpha}{m}}.\label{H15-5}
\end{align}
When $M'_j$ is the POVM element of all the outcomes corresponding to the correct computational result, we have
\begin{align}\label{p_incorrect}
& \big|\textrm{Tr} \big(M'_j \sigma - M_j|G\rangle \langle G|\big) \big| \nonumber\\
\le &
\big|\textrm{Tr} \big(U M'_j U^\dagger  - M_j \big) U \sigma U^\dagger \big|
+
\big|\textrm{Tr} M_j \big(U \sigma U^\dagger-|G\rangle \langle G|\big)\big|\nonumber \\
\le &
 {14 n \delta+ \frac{3 \alpha}{m}}.
\end{align}
Thus, choosing $m=O(n^4 \log n)$, we can achieve constant upper bound for the probability of accepting an incorrect output of the quantum computation
with significance level $\alpha$.
Connection between our protocol and interactive proof systems \cite{Babai1985,Goldwasser1989} is made explicit in Appendix \ref{app:int-proof}.

\section{Application to measurement-only blind quantum computation}
The above protocol may be applied to the scenario of measurement-only blind quantum computation \cite{morimae2013blind,Morimae2014verification,Hayashi2015} when the client does not trust the quantum devices performing the measurements.
Measurement-only blind quantum computation is a type of delegated quantum computation where the client with limited quantum power instructs a server to prepare a multipartite entangled state which is then sent to the client who performs single-qubit measurements that drive the computation.
This protocol is blind by construction, meaning the server cannot find out anything about the computation, and can be verified by stabiliser testing when the client trusts the measurement devices \cite{Hayashi2015}.
Measurement-only blind quantum computation was demonstrated experimentally in an optical setup in \cite{Greganti2016}.
Ability to quickly generate and measure quantum states is essential in any verification protocol therefore we believe that this setup shows great promise for implementing the self-guaranteed protocol in the near future.

Now we address the case when the measurement devices are not trusted.
We consider the client (Verifier) interacting with two servers, Prover 1 and Prover 2, where Prover 1 prepares the initial state and Prover 2 is used to measure the qubits and therefore test the state and the operation of the quantum devices.
There is a possibility that the noise in the initial state is correlated to the noise in the measurement devices.
To overcome this problem the Verifier asks Prover 1 to apply a random unitary $U(T):=
\left(
\begin{array}{cc}
\cos \frac{\pi T}{8} & - \sin \frac{\pi T}{8} \\
\sin \frac{\pi T}{8} & \cos \frac{\pi T}{8} 
\end{array}
\right)$ on every qubit, where $T$ is a uniform random number chosen from $\{0,1,\cdots,7\}$.
This technique of discrete twirling was also used to hide information about the initial state in the blind quantum computation protocols of \cite{Broadbent2009,fitzsimons2012unconditionally,fitzsimons2016private}.
We denote the full vector containing information about the applied local rotations by $\overrightarrow{T}$.
We assume that the order of measurements to be applied does not depend on the measurement outcomes in the computation.
This means that the qubits can be always measured in the same order regardless of the computation and only their basis depends on previous measurement outcomes.
The measurement process is composed of $n$ stages.
In the $i$-th stage of measurement, Verifier asks Prover 2 to measure the $i$-th qubit on each copy, taking into account the random unitary $U(T)$.
This discrete twirling removes the correlation among the sites on the color.
Therefore, the measurement devices on all sites and the state can be considered to be independent.

This protocol requires only independence among two parts, the part of generation of quantum states and the measurement devices.
In the language of interactive proof systems this would require Prover 1 and Prover 2 to be independent.
This requirement is usually enforced by considering Provers that are permitted to agree on a prior cheat strategy but are not allowed to communicate once the protocol commences.

The assumption of independence between the preparation stage and measurement stage is quite strong.
However it is necessary since Prover 1 could quite easily encode the information about the local random unitaries which Prover 2 could later use to his advantage.
This also highlights that the scenario considered in this verification scheme is not the usual one of protocols based on interactive proof systems where the Provers are assumed to be non-communicating.
Here Prover 1 and Prover 2 engage in one-way quantum communication necessitating our assumption that they are to a large degree honest.
On the other hand this assumption is natural in the context of verifying quantum technologies where the Provers are not assumed to be malicious and the only deviation from the verification protocol is caused by unexpected noise.
Similar less secure approach has been recently fruitfully used in \cite{jozsa2017} to efficiently verify adaptive Clifford circuits.

It is possible to enhance our protocol to the case where the Provers are considered malicious and are actively conspiring against the Verifier.
Assume that the Provers share a Bell pair for every qubit that Prover 1 is instructed to prepare.
The Verifier then asks Prover 1 to perform a two-qubit Bell measurement on the $i$-th qubit of a prepared graph state and its corresponding Bell-pair qubit, reporting the outcomes to the Verifier.
All the outcomes are denoted by vector $T'$.
The effect of these measurements is to teleport the prepared initial states from Prover 1 to Prover 2 up to a local unitary.
The Verifier then proceeds with self-testing protocol of Test \textbf{(4)} taking into account the local rotations $U(T+T')$.
Note that even if Prover 2 has access to the information about $T$ he cannot use it to cheat the Verifier as the vector $T'$ is also uniformly random and unknown to him.
By teleporting the copies of initial graph state from Prover 1 to Prover 2, the Verifier can check the computation without making any strong assumptions about the independence of the two Provers.
This addition to our protocol introduces only a multiplicative factor and does not affect the scaling of the overhead required by our protocol.

\section{Discussion}
The above analysis has been restricted to the case of three-colorable graph states.
In fact, the non-conflict condition can be relaxed to the case of graph states which are $k$-colorable as follows.
Firstly, we remember that our analysis can be divided into two parts, testing of the measurement basis and testing of the graph state.
The first part can be generalized as follows.
For each color $i=1, \ldots, k$, we divide the set of vertices with color $i$ into subsets $S_{i,1}, \ldots, S_{i,l_i}$ such that there is no common neighborhood for each subset $S_{i,j}$.
In this case, we can generalized the B-protocol as explained in Appendix \ref{app:multicolor}.
Then, applying this generalization to all colors in the protocol, we can extend the first part.
To realize the second part, for each color $i$, we measure non-$i$ color sites with $Z$ basis and check whether the outcome of measurement $X$ on the sites with color $i$ is the same as the predicted one.
We repeat this protocol for all colors.
Due to this construction we obtain the same analysis as three-colorable case when the numbers $l_1, \ldots, l_k$ are bounded.

For a computation on an $n$-qubit graph state the resources needed to achieve a constant upper bound for probability of accepting a wrong outcome scale as $O(n^4\log n)$.
This is the same scaling obtained in \cite{hajduvsek2015device} which raises an interesting open question of minimal overhead required to guarantee an outcome of quantum computation.
We have shown how our self-testing protocol can be applied to measurement-only blind quantum computation in the case when also the measurement devices cannot be fully trusted.

\section*{Acknowledgements}
MHayashi is partially supported by Fund for the Promotion of Joint International Research (Fostering Joint International Research) No. 15KK0007. This research is supported by the National Research Foundation, Prime Minister's Office, Singapore under its Competitive Research Programme (CRP Award No. NRF-CRP14-2014-02).

\appendix

\section{Significance level}\label{app:hhpo}
In statistical hypothesis testing, 
the significance level is a key concept for a given test $T$.
In general, we say that the statement $S$ holds with the significance level $\alpha$
when the probability of making an incorrect decision is less than $\alpha$ under this claim.
More precisely,
the probability of the following event is less than $\alpha$;
we claim the statement $S$ and the statement $S$ is incorrect.

If the statement $S$ is a property of the true initial state and measurement device,
this can be formulated of a simple form as follows.
We say that the true initial state and measurement device satisfy the property $S$ 
with the significance level $\alpha$
(or simply the property $S$ holds with the significance level $\alpha$)
when the test $T$ is passed and the following condition holds.
When the true initial state and measurement device do not satisfy the property $S$,
the probability to pass the test $T$ is smaller than $\alpha$.
Usually, there are many probability distributions of the whole system even when 
the true initial state and measurement device are fixed to not satisfy the property $S$
because we have several possibilities of the true initial state and measurement device.
This formulation is the conventional formulation in statistical hypothesis testing.
However, in the self-testing, we need to treat the case when 
the statement $S$ is not a property of the true initial state and measurement device,
which requires a more complicated formulation.

Assume that a certain property $S'$ implies that
the statement $S$ is correct with probability $\gamma$, which can be regarded as a kind of property of 
an initial state and a measurement device.
Now, we assume that the test $T$ is passed and we claim the statement $S$ as a result.
The case of making an incorrect decision is contained in the union of the following two cases.
One is the case when the property $S'$ does not hold.
The other is the case when the statement $S$ is not correct while  the property $S'$ holds.
When the property $S'$ holds with the significance level $\alpha$,
we can say that the statement $S$ holds with the significance level $\alpha+\gamma$.
This is because the probability of making an incorrect decision is less than 
the sum of the probabilities of the above two cases.

\section{Interval estimation with binomial distribution}\label{app:bi}
We consider how to verify the success probability of a binary system
by using sampling.
Assume that $n$ binary systems $X_1, \ldots, X_{n}$ take values in $\{1,0\}$.
We randomly choose $m$ systems and denote the sum by $X$.
We assume that the variables $X_1, \ldots, X_{n}$ independently obey the same distribution
$P(1)=p$ and $P(1)=1-p$.
We randomly choose one variable $X'$ from $n-m$ remaining systems.
Then, the variable $X'$ obeys the binary distribution with average $p$.
Now, we consider how to make a statement with respect to the average $p$
from the observed value $X$.

In the following, we denote the binomial distribution of $m$ trials with probability $p$
by $B_p$.
Given $p$, we define $x^+(p)$ as $\min\{ x| B_p( X \ge x )  \le \alpha \}$, which is often called the percent point with $\alpha$.
Then, when the observed $X$ satisfies $X \ge x^+(p_0)$,
we can say that the parameter $p$ is larger than $p_0$ with significance level $\alpha$.
When $m$ is sufficiently large,
$x^+(p_0) $ is approximated to $m p_0+ \sqrt{m}\Phi^{-1}(\alpha) \sqrt{p_0(1-p_0)}$.
Similarly, we define $x^-(p)$ as $\max\{ x| B_p( X \le x )  \le \alpha \}$.
Hence, when the observed $X$ satisfies $x^-(p_1) \ge X \ge x^-(p_0)$,
we can say that the parameter $p$ belongs to $(p_0,p_1 )$ with significance level $2\alpha$.

For a constant $a$
and a sufficiently large $m$, 
the value $\sqrt{\left(p_*+\frac{a}{\sqrt{m}}\right)\left(1-\left(p_*+\frac{a}{\sqrt{m}}\right)\right)}$
can be approximated to $\sqrt{p_*(1-p_*)}$.
We choose
$p_0=p_* - \frac{1}{\sqrt{m} }(\Phi^{-1}(\beta)+\Phi^{-1}(\alpha)) \sqrt{p_*(1-p_*)}$
and
$p_1=p_* + \frac{1}{\sqrt{m} }(\Phi^{-1}(\beta)+\Phi^{-1}(\alpha)) \sqrt{p_*(1-p_*)}$.
We define 
our test $T(m,p_*,\beta)$ 
($T^-(m,p_*,\beta)$)
by the condition that
the observed $X$ belongs to 
the interval
$[m p_* - {\sqrt{m} }\Phi^{-1}(\beta) \sqrt{p_*(1-p_*)},
m p_* + {\sqrt{m} }\Phi^{-1}(\beta) \sqrt{p_*(1-p_*)}
]$
(the half interval
$[m p_* + {\sqrt{m} }\Phi^{-1}(\beta) \sqrt{p_*(1-p_*)},+\infty)$).
We have the following two lemmas.

The following lemma guarantees the success probability when the test is passed
even when the system is maliciously prepared but the distribution is independent and identical, 
which relates to the soundness.

\begin{lemma}\Label{L1T}
When the test $T(m,p_*,\beta)$ ($T^-(m,p_*,\beta)$)
 is passed,
we can say that the parameter $p$ belongs to 
the interval
$[p_* - \frac{1}{\sqrt{m} }(\Phi^{-1}(\beta)+\Phi^{-1}(\alpha)) \sqrt{p_*(1-p_*)},
p_* + \frac{1}{\sqrt{m} }(\Phi^{-1}(\beta)+\Phi^{-1}(\alpha)) \sqrt{p_*(1-p_*)}]$ with significance level $\alpha$.
(the half interval
$[p_* - \frac{1}{\sqrt{m} }(\Phi^{-1}(\beta)+\Phi^{-1}(\alpha)) \sqrt{p_*(1-p_*)},+\infty)$ with significance level $\alpha$).
\end{lemma}

The following lemma guarantees that 
the test will be passed with high probability
when the system is well prepared to be 
the independent and identical distribution with success probability $p_*$, 
which relates to the completeness.

\begin{lemma}
Further, when the true parameter $p$ is $p_*$,
the test $T(m,p_*,\beta)$ ($T^-(m,p_*,\beta)$) is passed with probability 
$1-2\beta$ ($1-\beta$).
\end{lemma}

\section{Interval estimation with hyper-geometric distribution}\label{app:hyper}
In general, the variables $X_1, \ldots, X_{n}$ 
are not necessarily independent and identical.
Now, we consider such a general case.
Since $X$ is given as the sum of $m$ random samples,
the distribution of $X$ is given as 
$ \sum_{k=0}^{n} Q_K(k) P_{HG|n,m,k} $
with the distribution $Q_K$ of the hidden variable $K$ on $\{0, 1, \ldots, 2m+1\}$,
where the hypergeometric distribution $P_{HG|n,m,k}$ is given as
\begin{align}
P_{HG|n,m,k}(x):=\frac{{m \choose x} {n-m \choose k-x}
}{ {n \choose k}},
\end{align}
which has been employed in the security analysis on the quantum key distribution, for example \cite{shor2000simple,hayashi2006practical,hayashi2012concise}. 
Since $X'$ is a random choice from $n-m$ remaining systems,
when $K=k$ and $X=x$,
the variable $X'$ obeys the distribution $P_{HG|n-m,1,k-x}$, which equals the binary distribution
with average $\frac{k-x}{n-m} $.
In general,
when $X=x$,
the success probability of the binary distribution of $X'$ 
is $\sum_{k}P_{K|X}(k,x)\frac{k-x}{n-m} $.
Define the positive value
$c(\alpha,\beta)$ as
\begin{align}
c(\alpha,\beta)^2:=
\frac{\left(\frac{1}{2}+\frac{n}{n-m} \Phi^{-1}(\beta)\sqrt{p_*(1-p_*)}\right)^2}{\alpha}.
\end{align}
Now, we consider how to make a statement with respect to 
the success probability of 
the binary distribution of $X'$ from the observed value $X$.
The following lemma guarantees the success probability when the test is passed
even when the system is maliciously and the distribution is not independent nor identical, 
which relates to the soundness in the general case.

\begin{lemma}\Label{LL4}
Assume that $n= \kappa m+o(m)$.
When the test $T(m,p_*,\beta)$ is passed,
we can say that 
the success probability of the binary distribution of $X'$ 
belongs to the interval
$[p_* - \frac{c(\alpha,\beta)}{\sqrt{m} },
p_* + \frac{c(\alpha,\beta)}{\sqrt{m} }]$ 
with significance level $\alpha$.
\end{lemma}

We fix $\epsilon>0$
and 
choose 
$c:=\frac{\left(\frac{1}{2}+2 \Phi^{-1}(\beta)\sqrt{p_*(1-p_*)}\right)^2}{\alpha}$.
Due to Lemma \ref{LL4},
to guarantee 
that the success probability of a binary system belongs to the interval
$[p_* - \epsilon, p_* + \epsilon]$ 
with significance level $\alpha$,
we need to prepare $2\frac{c^2}{ \epsilon^2}$  samples 
and observe $\frac{c^2}{ \epsilon^2}$  samples.

Next, we define the test $T(m,0)$ by the condition that the observed $X$ equals $0$.
We have the following lemma,
\begin{lemma}\Label{LL5}
When the test $T(m,0)$ is passed,
we can say that 
the success probability of the binary distribution of $X'$ 
is less than $ \frac{1-\alpha}{m \alpha}$
with significance level $\alpha$.
\end{lemma}

We fix $\epsilon>0$.
Due to Lemma \ref{LL5},
to guarantee 
that the success probability of a binary system is less than
$\epsilon$ 
with significance level $\alpha$,
we need to prepare and observe $\frac{1-\alpha}{ \alpha\epsilon}$  samples.

Now, we discuss what kind of test will be used in this paper.
In the definition of $c(\alpha,\beta)$, the term $\frac{n}{n-m}$ appears.
If $n-m$ does not increase with the order $O(n)$, this term goes to infinity.
For example, when $m$ is a half of $n$, this term is $2$, which yields a useful application of Lemma \ref{LL4}.
Hence, when we need to verify that the binary variable has the success probability close to a certain non-zero value $p_*$, 
we use the test $T(m,p_*,\beta)$ with a half of observed values, i.e., $m=n/2$.
In contrast, when we need to verify that the binary variable has the success probability close to zero, 
we use all of observed values and employ the test $T(m,0)$.

\section{Proofs of Lemmas \ref{LL4} and \ref{LL5}}
\begin{proofof}{Lemma \ref{LL4}}
Step 1):
In this proof, 
the distribution $Q_K$ does not necessarily have a positive probability at one point.
To treat this case, we address the joint distribution $P_{XK}$ of $X$ and $K$
and the conditional distribution $P_{K|X}$.
When we observe $x$ as the outcome of the random variable $X$,
the success probability of the binary distribution of $X'$ is 
$p(x):= \sum_{k} P_{K|X}(k|x)\frac{k-x}{n-m}$.
Using $c_1:= \Phi^{-1}(\beta)\sqrt{p_*(1-p_*)}$
We define the function
\begin{align}
f(x):=
\left\{
\begin{array}{ll}
\left| p_*-p(x)\right| & \hbox{when } \left| p_*-\frac{x}{m}\right| \le \frac{c_1}{\sqrt{m}} \\
0 & \hbox{otherwise}.
\end{array}
\right.
\end{align}
The probability of making an incorrect decision is 
the probability of the event that
$f(X) > c(\alpha,\beta)$.
Therefore, it is sufficient to show that
this probability is less than $\alpha$.
As shown later, we have
\begin{align}
E_X f(X)^2
\le
\frac{1}{m}\left(\frac{1}{2}+\frac{n}{n-m} c_1\right)^2.\Label{HHT}
\end{align}
Then, Chebyshev inequality guarantees 
\begin{align}
P_X 
\left\{f(X) \ge \frac{c(\alpha,\beta)}{\sqrt{m} } \right\}
\le \alpha,
\end{align}
which is equivalent to the desired statement.

Step 2):
In the following, we show \eqref{HHT}.
First, we notice that
\begin{align}
E_X f(X)^2 & = \tilde{E}_X | p_*-p(X)|^2 \nonumber\\
& = \tilde{E}_X \left( E_{K|X} \left(p_*-\frac{K-X}{n-m}\right)\right)^2 \nonumber\\
& \le \tilde{E}_X \left( E_{K|X} \left(p_*-\frac{K-X}{n-m}\right)^2\right) \nonumber \\
& =  E_{K} \left( \tilde{E}_{X|K}\left(p_*-\frac{K-X}{n-m}\right)^2\right),\Label{TR1}
\end{align}
where $\tilde{E}_X Y$ expresses the expectation 
of $ 1_{[p_*-\frac{c_1}{\sqrt{m}}, p_*+\frac{c_1}{\sqrt{m}}  ]}(X) Y $
with respect to $X$
and $1_{[p_*-\frac{c_1}{\sqrt{m}}, p_*+\frac{c_1}{\sqrt{m}}  ]}$
is the indicator function on the set $[p_*-\frac{c_1}{\sqrt{m}}, p_*+\frac{c_1}{\sqrt{m}}  ]$.

For a given $k$, we set $c_2$ as
$ \left|\frac{k}{n}-p_*\right|=\frac{c_2}{\sqrt{m}}$.
When $\left| p_*-\frac{X}{m}\right| \le \frac{c_1}{\sqrt{m}}$, we have
\begin{align}
\frac{c_2-c_1}{\sqrt{m}}
\le
\left|\frac{X}{m}-\frac{k}{n}\right|
\le \frac{c_2+c_1}{\sqrt{m}}.\Label{TR4}
\end{align}
Then, we have
\begin{align}
\left|p_*-\frac{k-X}{n-m}\right| & \le \left|p_*-\frac{k}{n}\right|+\left|\frac{k}{n}-\frac{k-X}{n-m}\right| \nonumber\\
& = \left|p_*-\frac{k}{n}\right|+\frac{m}{n-m}\left|\frac{X}{m}-\frac{k}{n}\right| \nonumber \\
& \le \frac{c_2}{\sqrt{m}}+\frac{m}{n-m} \frac{c_1}{\sqrt{m}} \nonumber\\
& = \frac{1}{\sqrt{m}}\left(c_2+\frac{m}{n-m} c_1\right).\Label{TR2}
\end{align}
When $c_2> c_1$, using \eqref{TR4}, Chebyshev inequality guarantees 
\begin{align}
& P_{X|K=k} \left\{ \left| p_*-\frac{X}{m}\right| \le \frac{c_1}{\sqrt{m}}  \right\} \nonumber\\
& \le P_{X|K=k} \left\{ \frac{c_2-c_1}{\sqrt{m}} \le \left|\frac{X}{m}-\frac{k}{n}\right| \right\} \nonumber\\
& \le \frac{V_k}{m^2  \left(\frac{c_2-c_1}{\sqrt{m}}\right)^2} \nonumber\\
& \le \frac{1}{4(c_2-c_1)^2},
\Label{TR3}
\end{align}
where $V_k$ is the variance of $X$ and is calculated to be
$\frac{(n-m)m (n-k)k}{(n-1)n^2} \le \frac{m}{4}$.

The combination of \eqref{TR2} and \eqref{TR3} yields
\begin{align}
& \tilde{E}_{X|K=k}\left(p_*-\frac{k-X}{n-m}\right)^2 \nonumber\\
& \le \min \left(\frac{1}{4(c_2-c_1)^2}, 1\right)
\frac{1}{m}\left(c_2+\frac{m}{n-m} c_1\right)^2.
\end{align}
Since 
\begin{align}
& \max_{c_2: c_2\ge c_1}
\min \left(\frac{1}{4(c_2-c_1)^2}, 1\right)
\frac{1}{m}\left(c_2+\frac{m}{n-m} c_1\right)^2 \nonumber\\
& = \frac{1}{m}\left(\frac{1}{2}+\frac{n}{n-m} c_1\right)^2,
\end{align}
we have
\begin{align}
\tilde{E}_{X|K=k}
\left(p_*-\frac{k-X}{n-m}\right)^2
\le
\frac{1}{m}\left(\frac{1}{2}+\frac{n}{n-m} c_1\right)^2.
\Label{EET}
\end{align}
When $c_2 \le c_1$, we have
\begin{align}
\left|p_*-\frac{k-X}{n-m}\right|^2 & \le \frac{1}{m}\left(c_2+\frac{m}{n-m} c_1\right)^2 \nonumber\\
& \le \frac{1}{m}\left(\frac{1}{2}+\frac{n}{n-m} c_1\right)^2,
\end{align}
which implies \eqref{EET}.
The combination of \eqref{TR1} and \eqref{EET} implies 
\eqref{HHT}.
\end{proofof}

\begin{proofof}{Lemma \ref{LL5}}
Step 1):
We consider the case $n=m+1$.
While the true distribution is given as a probabilistic mixture of hypergeometric distributions,
it is sufficient to consider the mixture of the cases of $K=0,1$ because 
there is no possibility to pass the test when $K>1$.
Assume that $Q_K(1)=p$ and $Q_K(0)=1-p$.
The probability to pass is $\frac{1}{1+m}$ for $K=1$ and $1$ for $K=0$.
Hence, in general, the probability to pass is $1-p+\frac{p}{1+m}$.

When the test $T(m,0)$ is passed,
the success probability of the binary distribution of $X'$ 
is $ \frac{\frac{p}{1+m}}{1-p+\frac{p}{1+m}}$.
Therefore, when the probability to pass is greater than $\alpha$, i.e.,
$p \le \frac{(1+m)(1-\alpha)}{m}$
the success probability of the binary distribution of $X'$ 
is less than
$\frac{\frac{1}{1+m} \frac{(1+m)(1-\alpha)}{m}}{\alpha} 
=\frac{1-\alpha}{m\alpha}$.

Step 2):
We consider the general case. i.e., $n > m+1$.
Even in this case, if we focus on the observed variables $X_1, \ldots, X_m$
and the variable $X'$,
the behavior of $X$ can be written by a mixture of
hypergeometric distributions with $n=m+1$.
Hence, we obtain the desired statement.
\end{proofof}

\section{Bell state and Proof of Theorem 1}\label{app:Bell}
\subsection{Proof of Theorem 1 using a proposition}
In this appendix, we show Theorem 1 of the main body.
For this purpose, we focus on the observables
\begin{align}
\sX:=|1 \rangle \langle 0| +|0 \rangle \langle 1| ,
\quad
\sZ:=
|0 \rangle \langle 0| -|1 \rangle \langle 1| .
\end{align}
We consider the state  $(|0,+ \rangle +|1,- \rangle )/\sqrt{2}$ on the composite system $\cH_1 \otimes \cH_2$.
We also define
\begin{align}
\sA(0):= \frac{\sX+\sZ}{\sqrt{2}}, \quad \sA(1):= \frac{\sX-\sZ}{\sqrt{2}}. 
\end{align}

Here, instead of the ideal systems $\cH_1$ and $\cH_2$, we have the real systems $\cH_1'$ and $\cH_2'$.
Also, we assume that we can measure real observables
$\sX_1'$, $\sX_2'$,
$\sZ_1'$, $\sZ_2'$,
$\sA(0)_1'$, and
$\sA(1)_1'$.
Here, we choose the real systems $\cH_1'$ and $\cH_2'$ sufficiently large so that our measurements are the projective decompositions of these observables.

In the following, we prepare the real state $|\psi'\rangle $ on the composite system $\cH_1' \otimes \cH_2'$.
Then, we have the following proposition.
\begin{proposition}\label{T1}
When
\begin{align}
& \langle \psi'| \sA(0)_1'\left(\sX_2' +\sZ_2'\right)-\sA(1)_1'\left(\sX_2' - \sZ_2'\right) |\psi'\rangle \ge 2\sqrt{2} -\epsilon_1 \label{2-13-1} \\
& \langle \psi'| \sX_1' \sZ_2'  |\psi'\rangle \ge 1 -\epsilon_2 \label{ep4} \\
& \langle \psi'| \sZ_1' \sX_2'  |\psi'\rangle \ge 1 -\epsilon_3 \label{ep5} \\
& \langle \psi'| \sA(0)_1' (\sZ_2'+\sX_2')  |\psi'\rangle \ge \sqrt{2} -\epsilon_4 ,\label{ep6} \\
& \langle \psi'| \sA(1)_1' (\sZ_2'-\sX_2')  |\psi'\rangle \ge \sqrt{2} -\epsilon_4 ,\label{ep7} \\
& |\langle \psi'| \sX_1' \sX_2'+\sX_1' \sZ_2'  |\psi'\rangle| \le \epsilon_5 ,\label{2-13-2}
\end{align}
there exists a local isometry $U:\cH_1' \to \cH_1$ such that
\begin{align}
\left\|  U \sX_1' U^\dagger-\sX_1 \right\| & \le \delta_1 \label{2-13-3}\\
\left\|  U \sZ_1' U^\dagger-\sZ_1 \right\| & \le \delta_1 \\
\left\|  U \sA(0)_1' U^\dagger-\frac{\sX_1+\sZ_1}{\sqrt{2}} \right\| & \le \delta_2 \\
\left\|  U \sA(1)_1' U^\dagger-\frac{\sX_1-\sZ_1}{\sqrt{2}} \right\| & \le \delta_2 ,
\label{2-13-4}
\end{align}
where
$\delta_1:=
\sum_{j=1}^3
\hat{c}_j \epsilon_j^{\frac{1}{2}}$
and
$\delta_2:=
\sum_{j=1}^5
\bar{c}_j \epsilon_j^{\frac{1}{2}}
+ \sqrt{2}(\epsilon_2^{\frac{1}{4}}+ \epsilon_3^{\frac{1}{4}})$,
and $\hat{c}_j$ and $\bar{c}_j$ are constants.
\end{proposition}

This proposition will be shown in the next subsection.
Also, we prepare the following lemma.

\begin{lemma}\label{L02}
Given an acceptance probability $\beta$ and a significance level $\alpha$, 
there exist positive numbers $c'>0$ and $c''>c_1>0$ satisfying the following.
If the state $(|0,+\rangle + |1,-\rangle)/\sqrt{2}$ and measurement are prepared with no error, 
Test {\bf (2)} of the above $c_1$ is passed with probability $\beta$ (the acceptability).
Once all of the conditions in Step \textbf{(2-5)} with the above $c_1$ are satisfied, with the significance level $\alpha$, we can guarantee the conditions \eqref{2-13-1}-\eqref{2-13-2} with
$\epsilon_2,\epsilon_3= \frac{c'}{m}$ and $\epsilon_1=\frac{2c''}{\sqrt{m}},
\epsilon_4,\epsilon_5= \frac{c''}{\sqrt{m}}$.
\end{lemma}

\begin{proofof}{Theorem 1}
We choose three positive numbers $c'>0$ and $c''>c_1>0 $ as in Lemma \ref{L02}.
So, Lemma \ref{L02} guarantees the acceptability. 
Once all of the conditions in Step \textbf{(2-5)} with the above $c_1$ are satisfied, with the significance level $\alpha$, we can guarantee the conditions \eqref{2-13-1}-\eqref{2-13-2} with
$\epsilon_2,\epsilon_3= \frac{c'}{m}$ and $\epsilon_1=\frac{2c''}{\sqrt{m}},
\epsilon_4,\epsilon_5= \frac{c''}{\sqrt{m}}$.
Due to Proposition \ref{T1},
using a suitable isometry $U$, with the significance level $\alpha$, we can guarantee the conditions \eqref{2-13-3}-\eqref{2-13-4} with $\delta_1,\delta_2=O((\frac{1}{m })^{\frac{1}{4}})$, which yields the conditions \eqref{e5} and \eqref{e6}.

Therefore, due to Proposition \ref{T1},
using a suitable isometry $U$, with the significance level $\alpha$, we can guarantee the conditions \eqref{2-13-3}-\eqref{2-13-4} with $\delta_1,\delta_2=O(m^{-\frac{1}{4}})$, which is the desired argument.
\end{proofof}

\begin{proofof}{Lemma \ref{L02}}
The observables in the LHS of \eqref{ep4}-\eqref{ep5} take a deterministic value in the ideal state.
So, the acceptability for Eq. \eqref{e1} is automatically satisfied.
There exists a real number $c'$ satisfying the following condition.
To accept the tests $Av [X'_1 Z'_2] = 1$ and $Av [Z'_1 X'_2] = 1$ with more than probability $\alpha$, 
the conditions \eqref{ep4}-\eqref{ep5} of $\epsilon_2,\epsilon_3= \frac{c'}{m}$ need to hold.
Once these texts $Av [X'_1 Z'_2] = 1$ and $Av [Z'_1 X'_2] = 1$ are passed,
we can guarantee the conditions \eqref{ep4}-\eqref{ep5} of $\epsilon_2,\epsilon_3= \frac{c'}{m}$ 
with significance level $\alpha$.

On the other hand, the observables in the LHS of the remaining cases
\eqref{2-13-1}, \eqref{ep4}--\eqref{2-13-2} behave probabilistically even in the ideal state.
In order that the ideal state pass the tests \eqref{e2}, \eqref{e3}, and \eqref{e4}
with probability $\beta$,
the coefficient $c_1$  needs to be a constant dependent of $\beta$.
Then, dependently of $c_1$ and $\alpha$, there exists a real number $c''>c_1$ satisfying the following condition.
To accept the tests \eqref{e1}-\eqref{e4} with more than probability $\alpha$, 
the conditions \eqref{ep6}-\eqref{2-13-2} of $\epsilon_4,\epsilon_5= \frac{c''}{\sqrt{m}}$ need to hold.
Once the tests \eqref{e1}-\eqref{e4} are passed,
we can guarantee the conditions \eqref{ep6}-\eqref{2-13-2} of $\epsilon_4,\epsilon_5= \frac{c''}{\sqrt{m}}$ 
with significance level $\alpha$.
So, by choosing $\epsilon_1= 2 \epsilon_4$, 
\eqref{2-13-1} automatically holds.
\end{proofof}

\begin{remark}
Here we compare our overhead scaling with that in \cite{McKague2012}.
Their evaluation can be summarized as follows.
Let $\epsilon$ be the statistical error of observed variables like the quantities given in \eqref{2-13-1} -\eqref{2-13-2}.
\cite{McKague2012} focuses on the matrix norm of the difference between the ideal observables and the real observables like the quantities appearing in \eqref{2-13-3} -\eqref{2-13-4} and shows that these quantities are upper bounded by $O(\epsilon^{\frac{1}{4}})$.
However, \cite{McKague2012} does not discuss the relation between the error $\epsilon$ and the number of samples $m$.
Since their statistical errors are in the probabilistic case, the error $\epsilon$ is given as $O(1/\sqrt{m})$.
Hence, the above matrix norm is bounded by $O(m^{-1/8})$.
\end{remark}

\subsection{Proof of Proposition \ref{T1}}
To show Proposition \ref{T1}, we need several lemmas.

\begin{lemma}\label{L1}
When
\begin{align}
\langle \psi'| \sA(0)_1'\left(\sX_2' +\sZ_2'\right)-\sA(1)_1'\left(\sX_2' -\sZ_2'\right) |\psi'\rangle
& \ge 2\sqrt{2} -\epsilon_1
\end{align}
we have
\begin{align}
\|(\sX_2' \sZ_2' + \sZ_2' \sX_2' )
 |\psi'\rangle\|
& \le 2 \epsilon_1' ,
\end{align}
where
$\epsilon_1':= 2^{\frac{5}{4}}\epsilon_1^{\frac{1}{2}} $.
\end{lemma}

Lemma \ref{L1} follows from Theorem 2 of \cite{McKague2012}.

\begin{lemma}\label{L2}
When
\begin{align}
\langle \psi'| \sX_1' \sZ_2'  |\psi'\rangle
& \ge 1 -\epsilon_2 \label{24-1} \\
\langle \psi'| \sZ_1' \sX_2'  |\psi'\rangle
& \ge 1 -\epsilon_3 ,
\end{align}
we have
\begin{align}
\|  (\sX_1'- \sZ_2')  |\psi'\rangle \|
& \le \epsilon_2' \\
\| (\sZ_1'- \sX_2')  |\psi'\rangle \|
& \le \epsilon_3' ,
\end{align}
where
$\epsilon_j':=2^{\frac{1}{2}}
\epsilon_j^{\frac{1}{2}}$ for $j=2,3$.
\end{lemma}

\begin{proof}
Now, we make the spectral decomposition of $\sX_1' \sZ_2'$ as $\sX_1' \sZ_2'=P-(I-P)$, where $P$ is a projection. \eqref{24-1} implies that
$\langle \psi'| (I-P)  |\psi'\rangle
\le \frac{\epsilon_2}{2}$.
Schwarz inequality yields that
\begin{align}
\frac{1}{2} \|  (\sX_1'- \sZ_2' )  |\psi'\rangle\| & = \frac{1}{2} \| (I-\sX_1' \sZ_2' )  |\psi'\rangle\| \nonumber\\
& = \| (I-P)  |\psi'\rangle\| \nonumber \\
& \le \sqrt{\frac{\epsilon_2}{2}}.
\end{align}
Similarly, we obtain other inequalities.
\end{proof}

\begin{lemma}\label{L2-2}
When
\begin{align}
\langle \psi'| \sX_1' \sZ_2'  |\psi'\rangle
& \ge 1 -\epsilon_2 \\
\langle \psi'| \sZ_1' \sX_2'  |\psi'\rangle
& \ge 1 -\epsilon_3 ,\\
\langle \psi'| \sA(0)_1' (\sZ_2'+\sX_2')  |\psi'\rangle
& \ge 2 -\epsilon_4 ,\\
\langle \psi'| \sA(1)_1' (\sZ_2'-\sX_2')  |\psi'\rangle
& \ge 2 -\epsilon_4 ,\\
|\langle \psi'| \sX_1' \sX_2'+\sZ_1' \sZ_2'  |\psi'\rangle|
& \le \epsilon_5 ,
\end{align}
we have
\begin{align}
\left\|  \left(\sA(0)_1'-\frac{\sZ_2'+\sX_2'}{\sqrt{2}}\right)  |\psi'\rangle \right\|
& \le \epsilon_4' \label{T5-11-A}\\
\left\|  \left(\sA(1)_1'-\frac{\sZ_2'-\sX_2'}{\sqrt{2}}\right)  |\psi'\rangle \right\|
& \le \epsilon_4' ,\label{T5-11-B}
\end{align}
where
$\epsilon_4':=
\sqrt{ \sqrt{2}\epsilon_4
+ \frac{1}{2}\epsilon_5
+ \sqrt{\epsilon_2}+ \sqrt{\epsilon_3}}$.
\end{lemma}

\begin{proof}
Since
\begin{align}
\| (\sX_1' \sZ_2' -I) |\psi'\rangle \|
& \le 2 \sqrt{\epsilon_2} ,\\
\| (\sZ_1' \sX_2' -I) |\psi'\rangle \|
& \le 2 \sqrt{\epsilon_3} ,
\end{align}
we have
\begin{align}
& |\langle \psi'|\sZ_2'\sX_2'+\sX_2'\sZ_2'|\psi'\rangle - \langle \psi'|\sZ_2'\sZ_1'+\sX_2'\sX_1'|\psi'\rangle| \nonumber \\
\le& \|\langle \psi'|\sZ_2'  \sX_2'\| \| (I -\sZ_1' \sX_2')|\psi'\rangle \| \nonumber\\
+& \| \langle \psi'| \sX_2'\sZ_2' \| \| (I-\sX_1' \sZ_2' ) |\psi'\rangle \| \nonumber \\
\le& 2 \sqrt{\epsilon_2}+ 2 \sqrt{\epsilon_3}.
\end{align}
So, we obtain \eqref{T5-11-A} as follows.
\begin{align}
&\left\|  \left(\sA(0)_1'-\frac{\sZ_2'+\sX_2'}{\sqrt{2}}\right)  |\psi'\rangle \right\|^2 \nonumber  \\
=& \langle \psi'| {\sA(0)_1'}^2- \sqrt{2}\sA(0)_1'(\sZ_2'+\sX_2') \nonumber\\
+& \frac{{\sZ_2'}^2+{\sX_2'}^2+ \sZ_2'\sX_2'+\sX_2'\sZ_2'}{2}|\psi'\rangle \nonumber \\
=& 2-\langle \psi'| \sqrt{2}\sA(0)_1'(\sZ_2'+\sX_2') + \frac{\sZ_2'\sX_2'+\sX_2'\sZ_2'}{2} |\psi'\rangle\nonumber  \\
\le & 2- \sqrt{2}\langle \psi'|\sA(0)_1'(\sZ_2'+\sX_2')|\psi'\rangle \nonumber\\
+& \frac{1}{2} \langle \psi'|\sZ_2'\sZ_1'+\sX_2'\sX_1' |\psi'\rangle
+ \sqrt{\epsilon_2}+ \sqrt{\epsilon_3} \nonumber \\
\le &
2- 2 + \sqrt{2}\epsilon_4
+ \frac{1}{2}\epsilon_5
+ \sqrt{\epsilon_2}+ \sqrt{\epsilon_3} \nonumber \\
=&
 \sqrt{2}\epsilon_4
+ \frac{1}{2}\epsilon_5
+ \sqrt{\epsilon_2}+ \sqrt{\epsilon_3} .
\end{align}
Similarly, we have
\begin{align}
&\left\|  \left(\sA(1)_1'-\frac{\sZ_2'+\sX_2'}{\sqrt{2}}\right)  |\psi'\rangle \right\|^2 \nonumber  \\
\le& 2- \sqrt{2}\langle \psi'|\sA(1)_1'(\sZ_2'-\sX_2')|\psi'\rangle \nonumber\\
-& \frac{1}{2} \langle \psi'|\sZ_2'\sZ_1'+\sX_2'\sX_1' |\psi'\rangle
+ \sqrt{\epsilon_2}+ \sqrt{\epsilon_3} \nonumber \\
\le &
 \sqrt{2}\epsilon_4
+ \frac{1}{2}\epsilon_5
+ \sqrt{\epsilon_2}+ \sqrt{\epsilon_3} ,
\end{align}
which shows \eqref{T5-11-B}.
\end{proof}

\begin{lemma}\label{L3}
When
\begin{align}
\|(\sX_2' \sZ_2' + \sZ_2' \sX_2' )
 |\psi'\rangle\|
& \le 2 \epsilon_1' ,\\
\|  (\sX_1'- \sZ_2')  |\psi'\rangle \|
& \le \epsilon_2' ,\\
\| (\sZ_1'- \sX_2')  |\psi'\rangle \|
& \le \epsilon_3' ,\\
\left\|  \left(\sA(0)_1'-\frac{\sZ_2'+\sX_2'}{\sqrt{2}}\right)  |\psi'\rangle \right\|
& \le \epsilon_4' ,\\
\left\|  \left(\sA(1)_1'-\frac{\sZ_2'-\sX_2'}{\sqrt{2}}\right)  |\psi'\rangle \right\|
& \le \epsilon_4' ,
\end{align}
there exist local isometries $U_j:\cH_j' \to \cH_j$ for $j=1,2$
such that
\begin{align}
\|  U  |\psi' \rangle- |junk\rangle |\psi \rangle\| & \le \delta_1' ,\label{34}\\
\|  U \sX_1' |\psi' \rangle-\sX_1 |junk\rangle |\psi \rangle\| & \le \delta_1' ,\label{34-1}\\
\|  U \sZ_1' |\psi' \rangle-\sZ_1 |junk\rangle |\psi \rangle\| & \le \delta_1' ,\label{34-3}\\
\|  U \sX_2' |\psi' \rangle-\sX_2 |junk\rangle |\psi \rangle\| & \le \delta_1' ,\label{34-4}\\
\|  U \sZ_2' |\psi' \rangle-\sZ_2 |junk\rangle |\psi \rangle\| & \le \delta_1' ,\label{34-6}\\
\left\|  U \sA(0)_1' |\psi' \rangle-\frac{\sX_1+\sZ_1}{\sqrt{2}} |junk\rangle |\psi \rangle\right\| & \le \delta_2' ,\label{34-7}\\
\left\|  U \sA(1)_1' |\psi' \rangle-\frac{\sX_1-\sZ_1}{\sqrt{2}} |junk\rangle |\psi \rangle\right\| & \le \delta_2' ,\label{34-7B}
\end{align}
where $\delta_1':=\sum_{j=1}^3 c_j' \epsilon_j'$
and $\delta_2':=\sqrt{2}\delta_1'+ \epsilon_4' $
and $U:=U_2 U_1$.
\end{lemma}

\begin{proof}
We set the initial state on $\cH_1\otimes \cH_2$
to be $|0,+\rangle$.
Define
$U_1:=
(|0\rangle \langle 0|
+\sX_1'|1\rangle \langle 1|) \sH_1
(|0\rangle \langle 0|
+\sZ_1'|1\rangle \langle 1|) \sH_1 $
and
$U_2:=
(|0\rangle \langle 0|
+\sZ_2'|1\rangle \langle 1|) \sH_2
(|0\rangle \langle 0|
+\sX_2'|1\rangle \langle 1|) \sH_2 $,
where
$\sH:=
|+\rangle\langle 0|
+
|-\rangle\langle 1|
=
|0\rangle\langle +|
+
|1\rangle\langle -|$.

Hence, we have
\begin{align}
U  |\psi' \rangle & = \frac{1}{4}\left((I+\sZ_1')(I+\sX_2')|\psi'\rangle |0+\rangle \right. \nonumber\\
&+ \sZ_2' (I+\sZ_1')(I-\sX_2')|\psi'\rangle |0-\rangle \nonumber\\
&+ \sX_1'(I-\sZ_1')(I+\sX_2')|\psi'\rangle |1+\rangle \nonumber\\
&+ \left. \sX_1'\sZ_2'(I-\sZ_1')(I-\sX_2')|\psi'\rangle |1-\rangle\right).
\end{align}
When $|junk\rangle :=
\frac{\sqrt{2}}{4}(I+\sZ_1')(I+\sX_2')|\psi'\rangle $,
we have
\begin{align}
& U  |\psi' \rangle- |junk\rangle |\psi \rangle \nonumber\\
=& \frac{1}{4}( \sZ_2' (I+\sZ_1')(I-\sX_2')|\psi'\rangle |0-\rangle \nonumber\\
+& \sX_1'(I-\sZ_1')(I+\sX_2')|\psi'\rangle |1+\rangle \nonumber \\
+& (\sX_1'\sZ_2'(I-\sZ_1')(I-\sX_2')|\psi'\rangle \nonumber\\
-& (I+\sZ_1')(I+\sX_2')|\psi'\rangle )|1-\rangle).
\end{align}
We have
\begin{align}
& \| \sZ_2' (I+\sZ_1')(I-\sX_2')|\psi'\rangle |0-\rangle \| \nonumber \\
=& \| \sZ_2' (I+\sZ_1')(I-\sZ_1')|\psi'\rangle |0-\rangle \| \nonumber\\
+& \| \sZ_2' (I+\sZ_1')(\sZ_1'-\sX_2')|\psi'\rangle |0-\rangle \| \nonumber \\
\le & 2 \epsilon_3', \\
& \|
\sX_1'(I-\sZ_1')(I+\sX_2')|\psi'\rangle |1+\rangle
\| \nonumber \\
\le &
\|
\sX_1'(I-\sZ_1')(\sX_2'-\sZ_1')|\psi'\rangle |1+\rangle
\| \nonumber \\
\le & 2 \epsilon_3'.
\end{align}
Since
\begin{align}
& \sX_1'\sZ_2'(I-\sZ_1')(I-\sX_2')|\psi'\rangle \nonumber \\
=& \sZ_2'(I-\sX_2')\sX_1'(I-\sZ_1')|\psi'\rangle \nonumber\\
=& \sZ_2'(I-\sX_2')\sX_1'(I-\sX_2')|\psi'\rangle \nonumber\\
+& \sZ_2'(I-\sX_2')\sX_1'(\sX_2'-\sZ_1')|\psi'\rangle \nonumber \\
=& \sZ_2'(I-\sX_2')(I-\sX_2')\sZ_2'|\psi'\rangle \nonumber\\
+& \sZ_2'(I-\sX_2')(I-\sX_2')(\sZ_2'-\sX_1')|\psi'\rangle \nonumber\\
+& \sZ_2'(I-\sX_2')\sX_1'(\sX_2'-\sZ_1')|\psi'\rangle \nonumber \\
=& 2 \sZ_2'(I-\sX_2')\sZ_2'|\psi'\rangle+ 2 \sZ_2'(I-\sX_2')(\sZ_2'-\sX_1')|\psi'\rangle \nonumber\\
+& \sZ_2'(I-\sX_2')\sX_1'(\sX_2'-\sZ_1')|\psi'\rangle \nonumber \\
=& 2 \sZ_2'\sZ_2'(I+\sX_2')|\psi'\rangle-2 \sZ_2'(\sZ_2'\sX_2'+\sX_2'\sZ_2')|\psi'\rangle \nonumber\\
+& 2 \sZ_2'(I-\sX_2')(\sZ_2'-\sX_1')|\psi'\rangle \nonumber\\
+& \sZ_2'(I-\sX_2')\sX_1'(\sX_2'-\sZ_1')|\psi'\rangle \nonumber \\
=& (I+\sX_2')^2|\psi'\rangle -2 \sZ_2'(\sZ_2'\sX_2'+\sX_2'\sZ_2')|\psi'\rangle \nonumber\\
+& 2 \sZ_2'(I-\sX_2')(\sZ_2'-\sX_1')|\psi'\rangle \nonumber\\
+& \sZ_2'(I-\sX_2')\sX_1'(\sX_2'-\sZ_1')|\psi'\rangle \nonumber \\
=& (I+\sX_2')(I+\sZ_1')|\psi'\rangle ) +(I+\sX_2')(\sX_2'-\sZ_1')|\psi'\rangle \nonumber\\
-& 2 \sZ_2'(\sZ_2'\sX_2'+\sX_2'\sZ_2')|\psi'\rangle \nonumber \\
+& 2 \sZ_2'(I-\sX_2')(\sZ_2'-\sX_1')|\psi'\rangle \nonumber\\
+& \sZ_2'(I-\sX_2')\sX_1'(\sX_2'-\sZ_1')|\psi'\rangle,
\end{align}
we have
\begin{align}
& \|
(\sX_1'\sZ_2'(I-\sZ_1')(I-\sX_2')|\psi'\rangle
-(I+\sZ_1')(I+\sX_2')|\psi'\rangle )
\| \nonumber \\
\le &
(2+2) \epsilon_3'
+2 \epsilon_1'
+4 \epsilon_2'.
\end{align}
Thus,
\begin{align}
\|U  |\psi' \rangle- |junk\rangle |\psi \rangle\| & \le \frac{1}{4}(2 \epsilon_3'+2 \epsilon_3'+4 \epsilon_3'+2 \epsilon_1'+4 \epsilon_2') \nonumber\\
& = \frac{1}{2}(4 \epsilon_3'+ \epsilon_1'+2 \epsilon_2').
\end{align}
So, we obtain \eqref{34}.
Inequalities \eqref{34-1}-\eqref{34-6} can be shown by using the anti-commutation relation
and exchanging $\sX_1$, $\sZ_1$, $\sZ_1$ and
$\sX_2$, $\sX_2$, $\sZ_2$.
The coefficients $c_j$ for $\delta$ are given
by counting the number of these operations.

Now, we show \eqref{34-7}.
We have
\begin{align}
& \left\|  U \sA(0)_1' |\psi' \rangle -\frac{\sX_1+\sZ_1}{\sqrt{2}} |junk\rangle |\psi \rangle\right\| \nonumber \\
=& \left\|  U \sA(0)_1' |\psi' \rangle-\frac{\sZ_2+\sX_2}{\sqrt{2}} |junk\rangle |\psi \rangle\right\|\nonumber  \\
\le & \left\|  U \sA(0)_1' |\psi' \rangle-U \frac{\sZ_2'+\sX_2'}{\sqrt{2}} |\psi' \rangle \right\| \nonumber\\
+& \left\|U \frac{\sZ_2'+\sX_2'}{\sqrt{2}} |\psi'\rangle-\frac{\sZ_2+\sX_2}{\sqrt{2}} |junk\rangle |\psi \rangle\right\| \nonumber \\
=& \left\| \left(\sA(0)_1'-\frac{\sZ_2'+\sX_2'}{\sqrt{2}}\right)  |\psi'\rangle \right\| \nonumber\\
+& \frac{1}{\sqrt{2}} \left\| U \frac{\sZ_2'}{\sqrt{2}} |\psi' \rangle -\frac{\sZ_2}{\sqrt{2}} |junk\rangle |\psi \rangle\right\| \nonumber\\
+& \frac{1}{\sqrt{2}} \left\| U \frac{\sX_2'}{\sqrt{2}} |\psi' \rangle-\frac{\sX_2}{\sqrt{2}} |junk\rangle |\psi \rangle\right\| \nonumber \\
\le &
\epsilon_4'+ \sqrt{2}\delta_1'.
\end{align}
So, we obtain \eqref{34-7}.
In the same way, we can show \eqref{34-7B}.
\end{proof}

\begin{lemma}\label{L4}
The local isometries $U_j:\cH_j' \to \cH_j$ for $j=1,2$
satisfy
\begin{align}
\|  U  |\psi' \rangle- |junk\rangle |\psi \rangle\| & \le \delta_1' ,\\
\|  U \sX_1' |\psi' \rangle-\sX_1 |junk\rangle |\psi \rangle\| & \le \delta_1' ,\\
\|  U \sZ_1' |\psi' \rangle-\sZ_1 |junk\rangle |\psi \rangle\| & \le \delta_1' ,\\
\|  U \sX_2' |\psi' \rangle-\sX_2 |junk\rangle |\psi \rangle\| & \le \delta_1' ,\\
\|  U \sZ_2' |\psi' \rangle-\sZ_2 |junk\rangle |\psi \rangle\| & \le \delta_1' ,\\
\left\|  U \sA(0)_1' |\psi' \rangle
-\frac{\sX_1+\sZ_1}{\sqrt{2}} |junk\rangle |\psi \rangle\right\| & \le \delta_2' ,\\
\left\|  U \sA(1)_1' |\psi' \rangle
-\frac{\sX_1-\sZ_1}{\sqrt{2}} |junk\rangle |\psi \rangle\right\| & \le \delta_2' ,
\end{align}
for $U:=U_2 U_1$,
we have
\begin{align}
\|  U_1 \sX_1' U_1^\dagger-\sX_1 \| & \le 2\sqrt{2} \delta_1' \label{eq2}\\
\|  U_1 \sZ_1' U_1^\dagger-\sZ_1 \| & \le 2\sqrt{2} \delta_1' \\
\left\|  U_1 \sA(0)_1' U_1^\dagger-\frac{\sX_1+\sZ_1}{\sqrt{2}} \right\| & \le
\sqrt{2}(\delta_1'+\delta_2'), \\
\left\|  U_1 \sA(1)_1' U_1^\dagger-\frac{\sX_1-\sZ_1}{\sqrt{2}} \right\| & \le
\sqrt{2}(\delta_1'+\delta_2').
\end{align}
\end{lemma}

\begin{proof}
We have
\begin{align}
& U_1 \sX_1' U_1^\dagger |\psi\rangle |junk\rangle \nonumber \\
=& U_2 U_1 \sX_1' U_1^\dagger U_2^\dagger |junk\rangle |\psi\rangle \nonumber \\
=& U_2 U_1 \sX_1' U_1^\dagger U_2^\dagger U_2 U_1  |\psi'\rangle \nonumber\\
+& U_2 U_1 \sX_1' U_1^\dagger U_2^\dagger(|junk\rangle|\psi\rangle -U_2 U_1 |\psi'\rangle) \nonumber\\
=& U_2 U_1 \sX_1' |\psi'\rangle + U_2 U_1 \sX_1' U_1^\dagger U_2^\dagger (|junk\rangle|\psi\rangle -U_2 U_1 |\psi'\rangle) \nonumber\\
=& \sX_1 |\psi\rangle|junk\rangle + (U_2 U_1 \sX_1' |\psi'\rangle - \sX_1 |junk\rangle|\psi\rangle) \nonumber\\
+& U_2 U_1 \sX_1' U_1^\dagger U_2^\dagger (|junk\rangle|\psi\rangle -U_2 U_1 |\psi'\rangle).
\end{align}
Hence, we obtain
\begin{align}
\|
U_1 \sX_1' U_1^\dagger |\psi\rangle
-
\sX_1 |\psi\rangle\|
\le 2 \delta_1',
\end{align}
which implies that
\begin{align}
\|U_1 \sX_1' U_1^\dagger
-\sX_1 \|
\le 2\sqrt{2} \delta_1'.
\end{align}
So, we obtain \eqref{eq2}.
Similarly, we obtain other inequalities.
\end{proof}

\begin{proofof}{Proposition \ref{T1}}
Choose
$\delta_1=2\sqrt{2}\delta_1'$
and
$\delta_2=\sqrt{2}(\delta_1'+\delta_2')
=
\sqrt{2}((1+\sqrt{2})\delta_1'+
\epsilon_4')$.
Then, combining these lemmas, we obtain Proposition \ref{T1}.
\end{proofof}

\section{Proof of Theorem 2}\label{app:Th2}
To show  Theorem 2 of the main body, we prepare Lemma \ref{L01} as follows.

\begin{lemma}\label{L01}
Given an acceptance probability $\beta$ and a significance level $\alpha$, 
there exist positive numbers $c'''>c_4>0$ satisfying the following.
Here, we use the same $c'$ as the proof of Theorem 1.
If the state $|G\rangle$ and our measurements are prepared with no error, 
Test {\bf (4)} with $c_1=c_4 (\log n)^{1/2}$ is passed with probability $\beta$ (the acceptability).
Once all of the conditions in Step \textbf{(2-5)} 
with $c_1= c_4 (\log n)^{1/2}$ 
are satisfied in all sites, with the significance level $\alpha$, we can guarantee the conditions \eqref{2-13-1}-\eqref{2-13-2} in all sites with
$\epsilon_2,\epsilon_3= \frac{c'}{m}$ and $\epsilon_1=\frac{2c''' (\log n)^{1/2}}{\sqrt{m}},
\epsilon_4,\epsilon_5= \frac{c''' (\log n)^{1/2}}{\sqrt{m}}$.
\end{lemma}

We choose three positive numbers $c'>0$ and $c'''>c_4>0$ as in Lemma \ref{L01}.
So, Lemma \ref{L01} guarantees the acceptability. 
Once all of the conditions in Step \textbf{(2-5)} 
with $c_1= c_4 (\log n)^{1/2}$
are satisfied in all sites, 
with the significance level $\alpha$, we can guarantee the conditions \eqref{2-13-1}-\eqref{2-13-2} with
$\epsilon_2,\epsilon_3= \frac{c'}{m}$ and $\epsilon_1=\frac{2c''' (\log n)^{1/2}}{\sqrt{m}},
\epsilon_4,\epsilon_5= \frac{c''' (\log n)^{1/2}}{\sqrt{m}}$.
Due to Proposition \ref{T1},
using a suitable isometry $U$, with the significance level $\alpha$, we can guarantee the conditions \eqref{2-13-3}-\eqref{2-13-4} with $\delta_1,\delta_2=O((\frac{\log n}{m })^{\frac{1}{4}})$, which 
yields the conditions \eqref{H15-1} and \eqref{H15-2B}.

Eqs. \eqref{H15-3} can be shown as follows.
When $\Tr \sigma (I-P_i')\ge \frac{\alpha}{m} $,
we accept the stabilizer test with respect to $P_i$
with probability smaller than $\alpha$.
So, we can guarantee that $Tr \sigma (I-P_i')\le \frac{\alpha}{m} $ with significance level $\alpha$.

\begin{proofof}{Lemma \ref{L01}}
To accept the tests $Av [X'_1 Z'_2] = 1$ and $Av [Z'_1 X'_2] = 1$ 
in all sites with more than probability $\alpha$, 
the conditions \eqref{ep4}-\eqref{ep5} of $\epsilon_2,\epsilon_3= \frac{c'}{m}$ need to hold in all sites.
More precisely, the summand of $\epsilon_2$ and $\epsilon_3$ with respect to all sites
need to be $\frac{c'}{m}$, which is a stronger condition than the above condition. 
Once these tests $Av [X'_1 Z'_2] = 1$ and $Av [Z'_1 X'_2] = 1$ are passed in all sites,
we can guarantee the conditions \eqref{ep4}-\eqref{ep5} of $\epsilon_2,\epsilon_3= \frac{c'}{m}$ in all sites
with significance level $\alpha$.

The observables in the LHS of \eqref{ep4}-\eqref{ep5} 
and the stabilizer test take a deterministic value in the ideal state,
On the other hand, the observables in the LHS of the remaining cases
\eqref{2-13-1}, \eqref{ep6}--\eqref{2-13-2} behave probabilistically even in the ideal state.
Hence, for the acceptability, we need to care about the accepting probability only for \eqref{ep6}--\eqref{2-13-2} in all sites because 
\eqref{2-13-1} follows from \eqref{ep6} and \eqref{ep7}.
In order that the ideal state accepts all of the tests 
\eqref{e1}-\eqref{e4} in all sites, i.e., totally $3n$ tests,
with probability $\beta$, 
the coefficient $c_1$  needs to increase with respect to $n$.
For example, when we choose $c_1$ to be $c_4 (\log n)^{1/2}$ with a certain constant $c_4$,
the ideal state accepts these tests with probability $\beta$ in all sites due to the following reason.
To satisfy the above condition, we consider the case when 
the ideal state accepts each test of each site with probability $1-\frac{1-\beta}{4n}$,
which implies that the ideal state accepts all of $4n$ tests with more than probability $\beta$
because the test $|Av [X'_1 X'_2+ Z'_1 Z'_2]| \le \frac{c_1}{\sqrt{m}}$ is regarded as two tests.
For example, we focus on the test $Av [A(0)'_1(Z'_2+X'_2)] \ge \sqrt{2}- \frac{c_1}{\sqrt{m}}$.
Due to the central limit theorem, the accepting probability with the ideal state 
is approximated to
$\int_{\-\infty}^{\frac{c_1}{\sqrt{v}}} \frac{1}{\sqrt{2\pi}}e^{-\frac{x^2}{2}}dx$,
where $v$ is the variance of $A(0)'_1(Z'_2+X'_2)$.
By solving the condition
$\int_{\frac{c_1}{\sqrt{v}}}^{\infty} \frac{1}{\sqrt{2\pi}}e^{-\frac{x^2}{2}}dx =\frac{1-\beta}{4n}$,
we obtain $c_1=c_4 (\log n)^{1/2}$ with a certain constant $c_4$.
The remaining tests can be treated in the same way.
So, we can conclude that the above choice of $c_1$ guarantees the condition for the acceptance probability $\beta$.

Then, dependently of $c_4$ and $\alpha$, there exists a real number $c'''>c_4$ satisfying the following condition.
To accept the tests
$Av [A(0)'_1(Z'_2+X'_2)] \ge \sqrt{2}- \frac{c_4 (\log n)^{1/2}}{\sqrt{m}}$,
$Av [A(1)'_1(Z'_2-X'_2)]\ge \sqrt{2}- \frac{c_4 (\log n)^{1/2}}{\sqrt{m}}$, and
$|Av [X'_1 X'_2+ Z'_1 Z'_2]| \le \frac{c_4 (\log n)^{1/2}}{\sqrt{m}}$,
$Av [X'_1 Z'_2] = 1$, and $Av [Z'_1 X'_2] = 1$ in all sites with more than probability $\alpha (< \beta)$, 
we need to accept the tests
$Av [A(0)'_1(Z'_2+X'_2)] \ge \sqrt{2}- \frac{c_4 (\log n)^{1/2}}{\sqrt{m}}$,
$Av [A(1)'_1(Z'_2-X'_2)]\ge \sqrt{2}- \frac{c_4 (\log n)^{1/2}}{\sqrt{m}}$, and
$|Av [X'_1 X'_2+ Z'_1 Z'_2]| \le \frac{c_4 (\log n)^{1/2}}{\sqrt{m}}$
$Av [X'_1 Z'_2] = 1$ and $Av [Z'_1 X'_2] = 1$ in each sites with more than probability $\alpha (< \beta)$,
which implies that
the conditions \eqref{ep6}-\eqref{2-13-2} of $\epsilon_4,\epsilon_5= 
\frac{c_4 (\log n)^{1/2}+ c_1'}{\sqrt{m}}$ 
hold in each site with a certain constant $c_1'$.
Since $\frac{c''' (\log n)^{1/2}}{\sqrt{m}} > \frac{c_4 (\log n)^{1/2}+ c_1'}{\sqrt{m}}$,
the above condition implies  
the conditions \eqref{ep6}-\eqref{2-13-2} of $\epsilon_4,\epsilon_5= c''' (\frac{\log n}{m})^{1/2}$. 

Once these texts 
$Av [A(0)'_1(Z'_2+X'_2)] \ge \sqrt{2}- c_4 (\frac{\log n}{m})^{1/2}$,
$Av [A(1)'_1(Z'_2-X'_2)]\ge \sqrt{2}- c_4 (\frac{\log n}{m})^{1/2}$, and
$|Av [X'_1 X'_2+ Z'_1 Z'_2]| \le c_4 (\frac{\log n}{m})^{1/2}$ are passed in all sites,
we can guarantee the conditions 
the conditions \eqref{ep6}-\eqref{2-13-2} of $\epsilon_4,\epsilon_5= c''' (\frac{\log n}{m})^{1/2}$ 
in all sites with significance level $\alpha$.
So, by choosing $\epsilon_1= 2 \epsilon_4$, 
\eqref{2-13-1} automatically holds in all sites.
\end{proofof}

\section{Error of POVM element: Proof of Inequality \eqref{H15-4}}\label{app:errorPOVM}
Similarly to \cite{McKague2016}, we introduce
$n$ ideal trusted systems spanned by $|0\rangle ,|1\rangle$
while each untrusted system is spanned by $|1\rangle ,|-1\rangle$.
Let $U_j$ be a unitary on the trusted system.
Let $V_j$ be a unitary controlling
the $j$-th untrusted system by the trusted system, defined as follows.

The operators on the untrusted system are restricted to
$I$ and $s$ operators $\{\sD_{(i)}\}_{i=1}^s$
such that their eigenvalues are $1$ or $-1$
and $\| \sU \sD_{(i)} \sU^\dagger - \sD_{(i)} \|\le \delta$.
In the main text, $s=4$ and
$\{\sD_{(i)}\}_i=\{\sX,\sZ,\sA(0),\sA(1)\}$.
Then, we assume  $V_j$ has the form
$\sum_{k \in \bF_2^{n}}
|k\rangle \langle k| \sD_j(k)$,
where $\sD_j(k)$ is one of $I$ and $\{\sD_{(i)}\}_i$.
According to FIG.~7 of \cite{McKague2016},
we define $W_j:= U_j V_j W_{j-1},$ and $W_0=U_0$
and $\sU:= \sU_1\cdots \sU_n$.

\begin{proposition}[\protect{\cite[Lemma 6]{McKague2016}} with modification]
We have
\begin{align}
\| \sU W_j' \sU^\dagger -W_j \|
\le s j \delta .
\end{align}
\end{proposition}

\begin{proof}
We have
\begin{align}
\sU W_j' \sU^\dagger -W_j & = \sU U_j' \sU^\dagger V_j \sU W_{j-1}' \sU^\dagger -U_j V_j W_{j-1} \nonumber \\
& = (\sU U_j' \sU^\dagger - U_j )V_j \sU W_{j-1} '\sU^\dagger \nonumber\\
& + U_j V_j (\sU W_{j-1}' \sU^\dagger -W_{j-1} ).
\end{align}
By induction, it is enough to show
\begin{align}
\|
\sU U_j' \sU^\dagger - U_j
\|
\le s \delta .\Label{11-25-7B}
\end{align}
We have
\begin{align}
\sU U_j' \sU^\dagger - U_j & = \sU \sum_{k \in \bF_2^{n}} |k\rangle \langle k| \sD_j(k)' \sU^\dagger \nonumber\\
& - \sum_{k \in \bF_2^{n}}|k\rangle \langle k| \sD_j(k) \nonumber \\
& = \sum_{i=1}^s \sum_{\substack{k \in \bF_2^{n}: \\ \sD_j(k)=\sD_{(i)}}}  |k\rangle \langle k| (\sU \sD_{(i)}' \sU^\dagger -\sD_{(i)})
\end{align}
For $i$, we have
\begin{align}
& \Big \|\sum_{\substack{k \in \bF_2^{n}: \\ \sD_j(k)=\sD_{(i)}}}
|k\rangle \langle k| (\sU \sD_{(i)}' \sU^\dagger -\sD_{(i)})
\Big\| \nonumber\\
=& \Big\|\sum_{\substack{k \in \bF_2^{n}: \\ \sD_j(k)=\sD_{(i)}} }|k\rangle \langle k| \Big\| \cdot \Big\|\sU \sD_{(i)}' \sU^\dagger -\sD_{(i)}\Big\|\le \delta.
\end{align}
So, we have \eqref{11-25-7B}.
\end{proof}

Assume that our adaptive measurement is given as follows.
Once we obtain the measured outcomes $k_1, \ldots, k_{j-1}$,
we measure $\sD_j(k_1, \ldots, k_{j-1}) $ on the $j$-th system.
To discuss such an adaptive measurement,
we set the initial state $|+\rangle^{\otimes n}$ on the trusted system.
Choose $U_j$ as the application of the Hadamard operator $\sH$ on the $j$-th trusted system.
Then, we define
\begin{align*}
V_j & := \sum_{k_1, \ldots, k_{j-1}}|k_1, \ldots, k_{j-1},0\rangle \langle k_1, \ldots, k_{j-1},0| \nonumber\\
& + \sD_j(k_1, \ldots, k_{j-1})|k_1, \ldots, k_{j-1},1\rangle \langle k_1, \ldots, k_{j-1},1|.
\end{align*}
Then, we define the TP-CP map $\Lambda$ from
the untrusted $n$-qubit system to the trusted $n$-qubit system as
\begin{align}
\Lambda(\rho):= \Tr_{UT}
W_n \rho \otimes |+\rangle\langle +|^{\otimes n} W_n^\dagger,
\end{align}
where $\Tr_{UT}$ expresses the partial trace with respect to the untrusted system.
Due to the construction,
$\Lambda(\rho)$ is the same as the output distribution when the above adaptive measurement is applied.

\begin{proposition}[\protect{\cite[Corollary 2]{McKague2016}} with modification]\label{P3}
For any state $\rho$,
we have
\begin{align}
\| \sU \Lambda'(\sU^\dagger  \rho \sU)\sU^\dagger
- \Lambda( \rho ) \|_1
\le 2s n  \delta.\label{25-2}
\end{align}
\end{proposition}
Hence, when $M_i$ is a POVM element of an adaptive measurement, we have
\begin{align}
\| U M'_i U^\dagger -M_i \|
\le
\max_{\rho} \Tr ( U M'_i U^\dagger -M_i )\rho
\le 2 s n \delta ,\label{H15-4x}
\end{align}
which implies inequality \eqref{H15-4} of the main text by substituting $4$ for $s$.

\begin{proof}
We have
\begin{align}
& \sU W_j' \sU^\dagger (\rho \otimes |+\rangle\langle +|^{\otimes n}) \sU {W_j'}^\dagger \sU^\dagger-W_j(\rho \otimes |+\rangle\langle +|^{\otimes n})W_j^\dagger \nonumber \\
=& (\sU W_j' \sU^\dagger -W_j) (\rho \otimes |+\rangle\langle +|^{\otimes n}) \sU {W_j'}^\dagger \sU^\dagger \nonumber\\
+& W_j(\rho \otimes |+\rangle\langle +|^{\otimes n})(\sU {W_j'}^\dagger \sU^\dagger -W_j^\dagger).\label{16-6}
\end{align}
Also, we have
\begin{align}
&\|(\sU W_j' \sU^\dagger -W_j)(\rho \otimes |+\rangle\langle +|^{\otimes n})\sU {W_j'}^\dagger \sU^\dagger \|_1 \nonumber \\
\le & \|(\sU W_j' \sU^\dagger -W_j)\| \| (\rho \otimes |+\rangle\langle +|^{\otimes n}) \sU {W_j'}^\dagger \sU^\dagger \|_1 \nonumber\\
=& \|(\sU W_j' \sU^\dagger -W_j)\|\le s n \delta, \label{16-7}\\
& \| W_j
(\rho \otimes |+\rangle\langle +|^{\otimes n})
(\sU {W_j'}^\dagger \sU^\dagger -W_j^\dagger)\|_1 \nonumber \\
\le &
\| W_j
(\rho \otimes |+\rangle\langle +|^{\otimes n}) \|_1
\|(\sU {W_j'}^\dagger \sU^\dagger -W_j^\dagger)\| \nonumber \\
=&
\|(\sU {W_j'}^\dagger \sU^\dagger -W_j^\dagger)\| \le s n \delta.
\label{16-8}
\end{align}
Combining \eqref{16-6}, \eqref{16-7}, and \eqref{16-8},
we have
\begin{align}
& \| \sU \Lambda'(\sU^\dagger \rho \sU)\sU^\dagger-\Lambda(\rho)\|_1 \nonumber\\
\le& \|\sU W_j' \sU^\dagger(\rho \otimes |+\rangle\langle +|^{\otimes n})\sU {W_j'}^\dagger \sU^\dagger \nonumber\\
-& W_j(\rho \otimes |+\rangle\langle +|^{\otimes n})W_j^\dagger \|_1 \nonumber\\
\le& 2 s n \delta,
\end{align}
where the first inequality follows from the information processing inequality with respect to the trace of the untrusted system.
Hence, we obtain \eqref{25-2}.
\end{proof}

\section{Error in the initial state: Proof of inequality \eqref{H15-5}}\label{app:errorInitial}
In this section, we show a slightly stronger inequality than inequality \eqref{H15-5} of main text:
\begin{align}
\| U \sigma U^\dagger -|G\rangle \langle G| \|_1^2
\le 6 n \delta+ \frac{3 \alpha}{m}\label{2-16-10}
\end{align}
by assuming Inequalities (3)--(5) in Theorem 2.

Now, we have the relation
\begin{align}
& \| U \sigma U^\dagger -|G\rangle \langle G| \|_1 ^2 \nonumber\\
\stackrel{(a)}{\le}& 1 - \Tr \langle G| U \sigma U^\dagger |G\rangle = \Tr (I-|G\rangle \langle G|) U \sigma U^\dagger\nonumber  \\
\stackrel{(b)}{\le}& \Tr (I-P_1) U \sigma U^\dagger + \Tr (I-P_2) U \sigma U^\dagger \nonumber\\
+& \Tr (I-P_3) U \sigma U^\dagger ,
 \label{2-16-12}
\end{align}
where
$(a)$ follows from the relation between the trace norm and the fidelity \cite[(6.106)]{Hayashi2014}
and $(b)$ follows from the inequality
$I-|G\rangle \langle G|\le (I-P_1)+ (I-P_2)+ (I-P_3)$.

We can apply \eqref{H15-4x} with $s=2$ to $P_i$ because
$P_i$ is a POVM element of an adaptive measurement based on $X$ and $Z$.
So, we have
\begin{align}
& |\Tr (\sU (I-P_i')\sU^\dagger -(I-P_i) ) \sU \sigma \sU^\dagger| \nonumber\\
=& |\Tr (\sU P_i'\sU^\dagger -P_i) \sU \sigma \sU^\dagger| \nonumber\\
\le& \| \sU P_i'\sU^\dagger -P_i \| \le 2 n \delta.
\end{align}
Thus, Inequality (9) in Theorem 2 implies
\begin{align}
& \Tr (I-P_i) \sU \sigma \sU^\dagger \nonumber\\
=& \Tr ((I-P_i) -\sU (I-P_i')\sU^\dagger) \sU \sigma \sU^\dagger \nonumber\\
+& \Tr (I-P_i') \sigma \le 2 n \delta+\frac{\alpha}{m}.\Label{2-16-11}
\end{align}
The combination of \eqref{2-16-12}
and \eqref{2-16-11} yields \eqref{2-16-10}.

\section{Interactive proof system}\label{app:int-proof}
Although not explicitly stated in the main part, our protocol is an instance of an interactive proof system \cite{Babai1985,Goldwasser1989} for any language in \textbf{BQP} with a quantum prover (Bob) and a nearly-classical verifier (Alice) equipped with a random number generator.
More formally, for every language $L\in\textbf{BQP}$ and input $x$ there exists a poly($|x|$)-time verifier $V$ interacting with a poly($|x|$) number of quantum provers such that if $x\in L$, there exists a set of honest provers for which $V$ accepts with probability at least $c=2/3$ (completeness).
If $x\notin L$, then for any set of provers, $V$ accepts with probability at most $p_{incorrect}\leq s=1/3$ (soundness), where $p_{incorrect}$ is given by Eq.~(12) of the main part.

\section{Self testing with multi-colorable graph}\label{app:multicolor}
Now, we give a protocol for $k$-colorable graph state as follows.
For each color $i=1, \ldots, k$,
we divide the set $S_i$ of sites with color $i$ into subsets
$S_{i,1}, \ldots, S_{i,l_i}$ such that
there is no common neighborhood with non-$i$ color
for each subset $S_{i,j}$.
Then, as a generalization of B-protocol, we propose the $i$-protocol with the subset $S_{i,j}$ as follows.

\begin{description}
\item[(3-1)]
We prepare $8 m$ states $|G'\rangle$.

\item[(3-2)]
We measure $Z'$ on all sites of $S_i\setminus S_{i,j}$
for all copies.
Then, we apply $Z'$ operators on the remaining sites to correct applied
$Z'$ operators depending on the outcomes.

\item[(3-3)]
For all $a \in S_{i,j}$,
we choose a site $b_a \in N_a$.
Then, we measure $Z'$ on all sites of $(\cup_{t \neq i}S_t)\setminus \{b_a\}_{a \in S_{i,j}}$
for all copies.
Then, we apply $Z'$ operators on the remaining sites to correct applied
$Z'$ operators depending on the outcomes.

\item[(3-4)]
Due to the above steps,
the resultant state should be
$ \otimes_{a \in S_{i,j}}|\Phi'\rangle_{a b_a}$.
Then, we apply the self-testing procedure to
all of $\{|\Phi'\rangle_{a b_a}\}_{a \in S_{i,j}}$.
\end{description}
The above protocol verifies the measurement device on sites with $i$-th color.
Then, applying this generalization to all colors in the protocol, we can
extend the first part.

To realize the second part, for each color $i$, we measure non-$i$ color
sites with $Z$ basis and check whether the outcome of measurement $X$
on the sites with color $i$ is the same as the predicted one.
Then, we denote the projection corresponding to the passing event for this test by
$P_i$.
Hence, we have $\prod_{i=1}^k P_i= |G\rangle \langle G|$ because
only the state $|G\rangle$ can pass all of these tests.
Thus, applying this test for all colors,
we can test whether the state is the desired graph state.

Then, choosing $c_3$ to be $k+8(\sum_{i=1}^k l_i)$, we propose our self-testing protocol as follows,

\begin{description}
\item[(4-1)]
We prepare $c_3 m+1$ $n$-qubit states $|G'\rangle$.

\item[(4-2)]
We randomly divide the $c_3 m + 1$ copies into $c_3+1$ groups.
The first $c_3$ groups are composed of $m$ copies and the final group is composed of a single copy.

\item[(4-3)]
For the first $k$ groups, we apply the following test.
For the $i$-th group,
we measure $Z'$ on the sites with non-$i$ color and $X'$ on the sites with $i$-th color,
and check that the outcome of $X'$ measurements
is the same as predicted from the outcomes of $Z'$ measurements.

\item[(4-4)]
We run the $i$-protocol with $S_{i,j}$
for the $k+8(j-1+\sum_{i'=1}^{i-1}l_{i'})+1$-th
-$k+8(j+\sum_{i'=1}^{i-1}l_{i'})$-th groups.
Then, we check 8 conditions in Step {\bf (2-5)}.
We repeat this protocol for $j=1, \ldots, l_i$
and $i=1, \ldots, k$.
\end{description}

When we employ the above protocol for $k$-colorable case,
the difference from the 3-colorable case is only the number of samples.
we have the same analysis for the certification of computation result
as the 3-colorable case.

\bibliography{biblio}
\end{document}